\documentclass[titlepage, a4paper, 11pt]{article}

\usepackage[utf8]{inputenc}
\usepackage{amsmath}
\usepackage{amsfonts}
\usepackage{amssymb}
\usepackage{amsthm}
\usepackage{braket}
\usepackage{mathtools}
\usepackage{commath}
\usepackage{lmodern} 
\usepackage{diagbox}
\usepackage{enumerate}
\usepackage[shortlabels]{enumitem}
\usepackage[nodisplayskipstretch]{setspace}
\usepackage{courier}
\usepackage{graphicx}
\usepackage{geometry}
\usepackage{verbatim}
\usepackage[frozencache, cachedir=.]{minted} 

\everydisplay\expandafter{\the\everydisplay\setstretch{1}}

\newtheorem{theorem}{Theorem}[section]
\newtheorem{corollary}{Corollary}[theorem]
\newtheorem{lemma}[theorem]{Lemma}
\newtheorem{definition}[theorem]{Definition}

\DeclareMathOperator{\Tr}{Tr}

\title{Numerical Methods for Quantum Spin Dynamics}
\author{Danny Goodacre\\[0.6cm]{\small Supervisor: Dr. Pranav Singh}}
\date{April 2022}

\begin{document}

\maketitle

\pagenumbering{roman}
\section*{Abstract}

This report is concerned with the efficiency of numerical methods for simulating quantum spin systems. We aim to implement an improved method for simulation with a time-dependent Hamiltonian and behaviour which displays chirped pulses at a high frequency.

Working in the density matrix formulation of quantum systems, we study evolution under the Liouville-von Neumann equation in the Liouvillian picture, presenting analysis of and benchmarking current numerical methods. The accuracy of existing techniques is assessed in the presence of chirped pulses.

We discuss the Magnus expansion and detail how a truncation of it is used to solve differential equations. The results of this work are implemented in the Python package MagPy to provide a better error-to-cost ratio than current approaches allow for time-dependent Hamiltonians. Within this method we also consider how to approximate the computation of the matrix exponential.

\subsection*{Motivation}

The spin of a particle is an entirely quantum mechanical quantity
that has no classical analogue. Accurately predicting the dynamics of
spins systems is of great importance in applications such as quantum
computing and nuclear magnetic resonance imaging (NMR/MRI) \cite{spins, spins2}. Our work is focused on systems displaying highly oscillatory chirped pulses (HOCP), which are prevalent in nuclear spin dynamics.

Simulating larger systems (2-5 spins) is common in spectroscopy and quickly becomes prohibitively expensive using current techniques, motivating the development of a method which is more computationally efficient.

\subsection*{Outline}

The report is organised as follows: the first section provides background in the mathematics and theoretical concepts, presenting relevant equations and proving theorems which are used in our application throughout. 

Within the first section we introduce the formalism that we use to represent quantum spin systems and discuss the equation underpinning their evolution over time. We also detail the structure of the Hamiltonian which we consider--for both single and multiple spins, interacting and non-interacting---and introduce notation for concisely writing such Hamiltonians. 

Next, we discuss the Liouvillian picture of quantum mechanics, giving us a  form to which we can readily apply various numerical methods. We prove some useful results of the Liouvillian superoperator which will be later used in our implementation.

Finally, we discuss how to measure the expectation values and spin components of spin systems and visualise these in 3D space where possible.

In the second section, we discuss the Python package QuTiP and review current numerical methods for solving initial value problems for ordinary differential equations and both polynomial and iterative techniques for approximating the matrix exponential. 

For each method in this section we discuss their shortcomings regarding highly oscillatory systems, such as what properties they do or do not conserve, and how the error in the approximation is bounded depending on the qualities of the system.

The third section begins with defining the Magnus expansion for time-dependent first-order homogeneous linear differential equations. Truncating the expansion, we reformulate the remaining terms to be applied more effectively to the Hamiltonians with which we are concerned. An example of our implementation in Python is provided with instructions for its use.

The fourth section is primarily concerned with comparing the effectiveness of our Magnus-based approach against current techniques. We analyse the rates of convergence for both our approach and others, and also consider the error-to-cost ratio for different approximations for evaluating the integals involved in the Magnus expansion.

\pagebreak

\newgeometry{top=4.4cm}
\tableofcontents
\restoregeometry

\pagebreak

\pagenumbering{arabic}

\section{Background}
In this section we discuss the mathematical representation and properties of quantum spin systems; we define the Hamiltonian of a system, which---alongside a given initial condition---determines the evolution of the system over time; and we discuss the resulting differential equation describing the dynamics of a quantum spin system and reformulate to the Liouvillian picture to apply existing numerical methods. Furthermore, throughout this section we prove properties of the Liouvilillan superoperator and introduce some new notation.

\subsection{Density Matrices}
We represent quantum states by density matrices, denoted $\rho(t)$. For a $n$ spin system, it is a complex-valued $2^n\times 2^n$ matrix. Density matrices are Hermitian, positive semi-definite, and have unit trace
$$\rho^\dagger = \rho,\;\;\;\rho\geq0,\;\;\;\Tr\{\rho\} = 1.$$
The initial condition of the density matrix is denoted as $\rho_0$. Later in this section we discuss how to use the density matrix to determine properties of the system.

For a system of $n$ spins that are not interacting, the density matrix can be written as 
\begin{equation}
\rho(t) = \rho_1(t) \otimes \cdots \otimes \rho_n(t),
\label{nonIntDensityMatrix}
\end{equation}
where $\rho_k(t)$ is the density matrix for each respective spin and $\otimes$ is the Kronecker product. When there is interaction between the spins, the density matrix cannot be written as such.

\subsection{The Hamiltonian}
The Hamiltonian for an $n$ spin system is an $2^n \times 2^n$ complex-valued Hermitian (self-adjoint) matrix. For such a Hamiltonian,
$$iH\in\mathfrak{su}(2^n),$$
where $\mathfrak{su}(2^n)$ is the Lie algebra of $\mathrm{SU}(2^n)$ and consists of $n\times n$ skew-Hermitian matrices with trace zero, with the operator
$$[A, B] = AB - BA,\;\;\;\;A,\,B\in\mathfrak{su}(2^n),$$
called the commutator \cite{su}. The space $\mathfrak{su}(2^n)$ is closed under commutation of any two of its elements, meaning the commutator as a matrix is also skew-Hermitian
$$[A, B]^\dagger = -[A, B].$$

\subsubsection{Single Spin Hamiltonian}
The space $\mathfrak{su}(2)$, has dimension 3 and a basis formed of the skew-Hermitian matrices
$$i\sigma_x,\;\;i\sigma_y,\;\;i\sigma_z,$$
where
$$\sigma_x = \begin{pmatrix}0 & 1 \\ 1 & 0\end{pmatrix},\;\;\sigma_y = \begin{pmatrix}0 & -i \\ i & 0\end{pmatrix},\;\;\sigma_z = \begin{pmatrix}1 & 0 \\ 0 & -1\end{pmatrix},$$
called the Pauli spin matrices \cite{su}. These matrices satisfy the following commutation relations,
$$ [\sigma_x,\sigma_y] = 2i\sigma_z,\;\;[\sigma_y,\sigma_z] = 2i\sigma_x,\;\;[\sigma_z,\sigma_x] = 2i\sigma_y.$$
Given two real functions and a constant,
$$f:\mathbb{R}\to\mathbb{R},\;\;\;g:\mathbb{R}\to\mathbb{R},\;\;\;\Omega\in\mathbb{R},$$
we define
$$H: \mathbb{R}\to\mathbb{C}^{2\times2}$$
$$H(t) = f(t)\sigma_x + g(t)\sigma_y + \Omega\sigma_z,$$
which is the form of a single spin Hamiltonian used throughout this report. As the Pauli matrices are Hermitian and $\mathfrak{su}(2)$ is a linear subspace, this form of Hamiltonian is still Hermitian. Note that for a time-independent Hamiltonian we simply take $f$ and $g$ to be constant functions.

\subsubsection{Multi-Spin Hamiltonian}
When the system contains more than one spin, it can be either interacting or non-interacting. Suppose the system contains $n$ spins. For each spin, we define a Hamiltonian
$$H_j(t) = f_j(t)\sigma_x + g_j(t)\sigma_y + \Omega_j\sigma_z,$$
for $j \in \{1,\ldots,n\}$. Then, for the non-interacting case, the Hamiltonian is
defined as\begin{align*}
    H_0(t) = \,\,&H_1(t) \otimes I \otimes \cdots \otimes I \\
    &+ I \otimes H_2(t) \otimes \cdots \otimes I \\
    &+ \cdots \\
    &+ I \otimes \cdots \otimes I \otimes H_n(t),
\end{align*}
where $I$ is the $2\times2$ identity matrix. If the spins interact, there is an additional term in the Hamiltonian,
$$H_J\in\mathbb{C}^{2^n \times 2^n},$$ 
which is time-independent and Hermitian. Thus the Hamiltonian for an interacting multi-spin system is
$$H(t) = H_0(t) + H_J.$$
Therefore, in order to specify an $n$ spin Hamiltonian, we require
$$\left(f_j(t),\,g_j(t),\,\Omega_j\right),\;\;\;j\in\{1,\ldots,n\},$$
and the constant matrix $H_J$ if the system is interacting.

\subsubsection{Highly Oscillatory Chirped Pulse Systems}
A system with which we are primarily concerned is one where each spin displays highly oscillatory chirped pulses (HOCP) in its evolution. We define such a system for a single spin  as follows: let
\begin{align*}
    e(t) &= \beta \exp\bigg(\frac{-(t-10)^8}{10^7}\bigg), \\
    w(t) &= \exp\bigg(i\,\gamma\,(t-10)^2\bigg).
\end{align*}
Then,
\begin{align*}
    f(t) &= e(t)\,\operatorname{Re}\big\{w(t)\big\}, \\
    g(t) &= e(t)\,\operatorname{Im}\big\{w(t)\big\},
\end{align*}
where $\operatorname{Re}$ and $\operatorname{Im}$ denote the real and imaginary parts, respectively. Throughout this report we define such a system by specifying the values of $\beta$, $\gamma$, and $\Omega$ for each spin's Hamiltonian, an initial condition $\rho_0$, and a constant matrix $H_J$ if the spins are interacting. 

\begin{figure}[H]
    \centering
    \includegraphics[width=\textwidth]{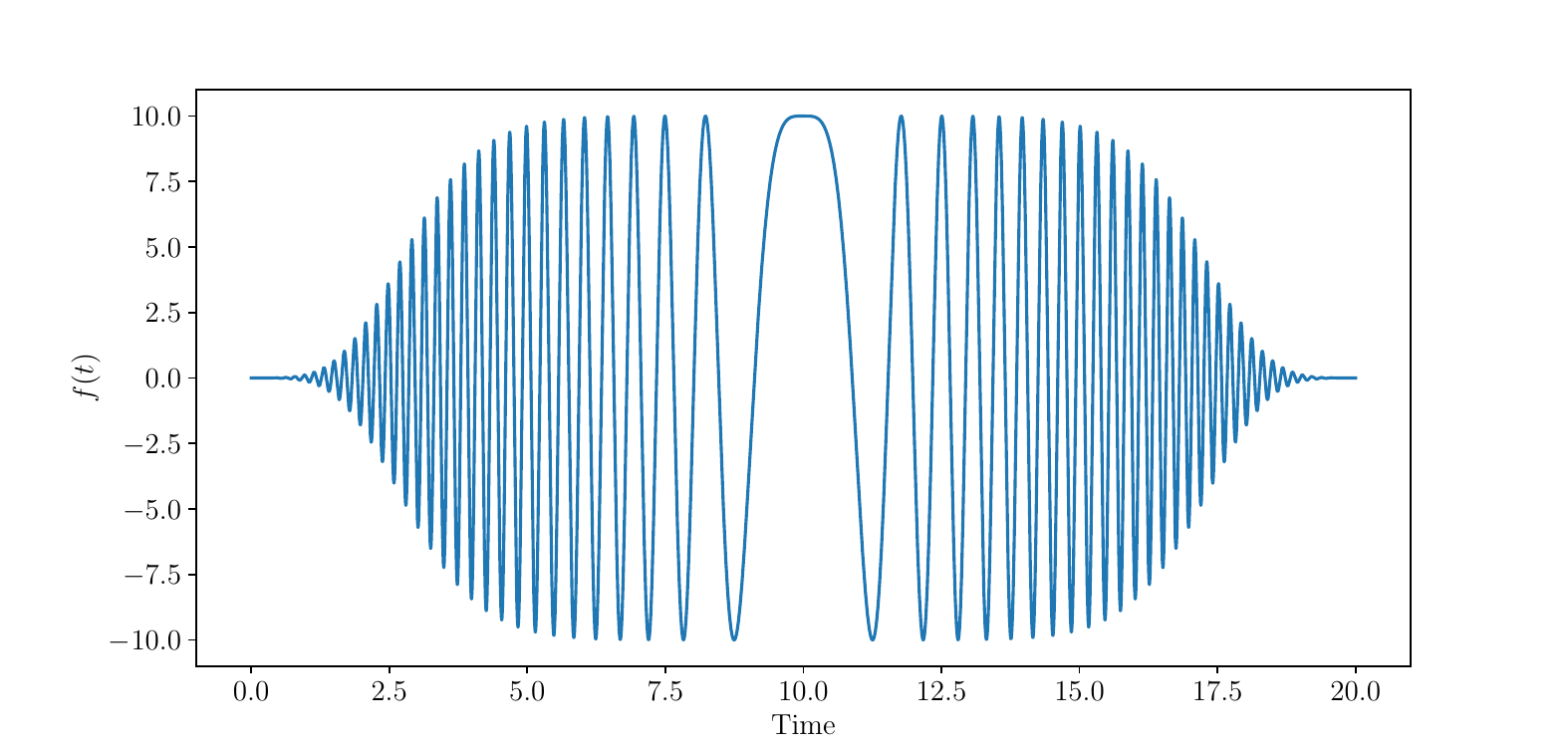}
    \caption{The function $f$ with $\beta=10$ and $\gamma=2$.}
    \label{hocpExample}
\end{figure}

\subsubsection{Notation}
Since the notation for a Hamiltonian is laborious to write, we introduce the following operators $\mathbf{I}$ and $\mathbf{S}$.

\begin{definition}
\label{I}
    $$\mathbf{I}_j\{A\} := I \otimes \cdots \otimes I \otimes A \otimes I \otimes \cdots \otimes I,$$
    where $A$ is in the $j$th position. The number of terms in the product is inferred from the number of spins in the system.
\end{definition}
\noindent
This operator is linear due to the linearity of the Kronecker product. Using this, we write the Hamiltonian as 
$$H(t) = \sum_{j=0}^n \mathbf{I}_j\left\{H_j(t)\right\} + H_J$$
When $A$ is a Pauli matrix, we simplify further,
$$\mathbf{I}_j^a := \mathbf{I}_j\{\sigma_a\},\;\;\;a\in\{x,y,z\}.$$
Noting that $H_j(t) = f_j(t)\,\sigma_x + g_j(t)\,\sigma_y + \Omega_j\,\sigma_z,$ we apply the linearity of the Kronecker product to get $$\mathbf{I}_j\left\{H_j(t)\right\} = f_j(t)\,\mathbf{I}_j^x + g_j(t)\,\mathbf{I}_j^y + \Omega_j\,\mathbf{I}_j^z.$$
Therefore the full interacting $n$ spin Hamiltonian can be written as 
\begin{equation}
    H(t) = \sum_{j=1}^n \left(f_j(t)\,\mathbf{I}_j^x + g_j(t)\,\mathbf{I}_j^y + \Omega_j\,\mathbf{I}_j^z\right) + H_J.
    \label{eqn:Hamiltonian}
\end{equation}
This leads us to the following,
\begin{lemma}
\label{Integral of Hamiltonian}
    \begin{equation*}
        \int_a^b H(t)\,dt = \sum_{j=1}^n \left(\int_a^b f_j(t)\,dt\,\mathbf{I}_j^x + \int_a^b g_j(t)\,dt\,\mathbf{I}_j^y + \Omega_j(b-a)\,\mathbf{I}_j^z\right) + H_J(b-a)
    \end{equation*}
\end{lemma}
\begin{proof}
Follows directly from definition of $H(t)$ and linearity of integration.
\end{proof}

\begin{lemma}[Generalised mixed-product]
\label{Generalised mixed product}
    Given matrices $A_j,B_j\in\mathbb{C}^{n\times n},\;\\j\in\{1,\ldots,k\},$
    $$(A_1 \otimes A_2 \otimes \cdots \otimes A_k)(B_1 \otimes B_2 \otimes \cdots \otimes B_k) = A_1 B_1 \otimes A_2 B_2 \otimes \cdots \otimes A_k B_k.$$
\end{lemma}
\begin{proof}
Induct on $k$---c.f. Broxson \cite{BroxsonKron}.
\end{proof}

\begin{definition}
\label{S}
    Given matrices $A, B\in\mathbb{C}^{n\times n}$, let
    $$\mathbf{S}_{i,\,j}\left\{A,B\right\} := I \otimes \cdots \otimes I \otimes A \otimes I \otimes \cdots \otimes I \otimes B \otimes I \otimes \cdots \otimes I,$$
    where $A$ is in the $i$th position and $B$ is in the $j$th position. When $i=j$ we write
    \begin{align*}
        \mathbf{S}_{j,\,j}\left\{A,B\right\} &= I \otimes \cdots \otimes I \otimes AB \otimes I \otimes \cdots \otimes I \\
        &= \mathbf{I}_{j}\left\{AB\right\}.
    \end{align*}
\end{definition}

\begin{lemma}
\label{S properties}
    Given matrices $A, B\in\mathbb{C}^{n\times n}$,
    \begin{enumerate}[1.]
        \item $\mathbf{S}_{i,\,j}\left\{A,B\right\} = \mathbf{I}_i\left\{A\right\}\,\mathbf{I}_j\left\{B\right\},$
        \item When $i\neq j$,\, $\mathbf{S}_{i,\,j}\left\{A,B\right\} = \mathbf{S}_{j,\,i}\left\{B,A\right\},$
        \item $\mathbf{S}_{j,\,j}\{A,B\} = \mathbf{S}_{j,\,j}\{B,A\} \Longleftrightarrow AB=BA.$
    \end{enumerate}
\end{lemma}
\begin{proof}
    \begin{enumerate}[1.]
        \item Apply the generalised Kronecker product to $\mathbf{I}_i\left\{A\right\}\,\mathbf{I}_j\left\{B\right\},$
        \item This follows directly from the definition.
        \item This also follows directly from the definition.
    \end{enumerate}
\end{proof}
\noindent
To conclude, the following conventions will be maintained throughout:
\begin{itemize}
    \item $H_j(t),\;j\in\mathbb{N}$ refers to a function of the form
    $$f_j(t)\,\sigma_x + g_j(t)\,\sigma_y + \Omega_j\,\sigma_z,$$
    \item $H_0(t)$ is
    $$\sum_{j=1}^n \mathbf{I}_j\left\{H_j(t)\right\},$$
    \item and $H(t)$ is
    $$H_0(t) + H_J.$$
\end{itemize}

\subsection{The Liouville-von Neumann Equation}
Given a Hamiltonian, $H(t)$, the corresponding dynamics of the spin system's density matrix, $\rho(t)$, are determined by the following ordinary differential equation,
\begin{equation*}
    \frac{\partial\rho(t)}{\partial t} = -i\left[H(t),\,\rho(t)\right],
\end{equation*}
called the Liouville-von Neumann equation \cite{Mazzi}. Along with an initial condition, $\rho_0$, we get a matrix-valued initial value problem. Using our reduced notation, the initial density matrix for an $n$ particle non-interacting system is written as 
$$\rho_0 = \sum_{j=1}^n \mathbf{I}_j\left\{\rho_j\right\}.$$

\subsubsection{The Liouvillian Picture}
In order to more readily apply existing numerical methods to the Liouville-von Neumann equation, we change from a Hamiltonian framework to a Liouvillian framework. This is achieved through use of the Liouvillian superoperator,
$$L:\mathbb{C}^{n\times n}\to\mathbb{C}^{n^2\times n^2}$$
$$L\{X\} = I_n \otimes X - X^T \otimes I_n,$$
and column-major vectorisation
$$\text{vec}:\mathbb{C}^{n\times n} \to \mathbb{C}^{n\times1},$$
denoted $\text{vec}(X)$. From \cite{MacedoVec} we have
$$\text{vec}(XY) = (I \otimes X)\text{vec}(Y) = (Y^T \otimes I)\text{vec}(X),$$
and so, denoting $\text{vec}(\rho(t))$ by $\mathbf{r}(t)$, the Liouvillian-von Neumann equation is transformed into
\begin{equation}
    \frac{\partial\mathbf{r}(t)}{\partial t} = -iL\{H(t)\}\,\mathbf{r}(t),
    \label{veclvn}
\end{equation}
which is in the form of a standard vector-valued initial value problem:
$$\mathbf{r}'(t) = A(t)\,\mathbf{r}(t),\;\;\mathbf{r}(t_0) = \mathbf{r}_0,$$
where $\mathbf{r}_0 = \text{vec}(\rho_0).$ It is in this form to which we apply our numerical methods in later sections.

\subsubsection{Properties of the Liouvillian}

The main result here is Theorem \ref{Double Integral of Commutator of Liouvillian} and is vital in the later implementation of the Magnus expansion. Theorem \ref{Conservation of Hermitivity} is also useful as it shows that the Liouvillian preserves a useful property of the Hamiltonian.

\begin{theorem}[Convservation of Hermitivity]
\label{Conservation of Hermitivity}
    For a Hermitian matrix $H$, $L\left\{H\right\}$ is Hermitian.
\end{theorem}
\begin{proof}
    For a Hermitian matrix $H$, consider $$L\left\{H\right\} = I \otimes H - H^T \otimes I.$$
    Taking the conjugate transpose of both sides gives
    $$L\left\{H\right\}^\dagger = \left(I \otimes H\right)^\dagger - (H^T \otimes I)^\dagger.$$
    Since conjugate transposition is distributive over the Kronecker product and $H$ is Hermitian, we then get
    $$L\left\{H\right\}^\dagger = I \otimes H - H^T \otimes I.$$
    Thus $L\left\{H\right\}$ is Hermitian.
\end{proof}

\begin{lemma}
\label{Integral of liouvillian}
    For a matrix $X\in\mathbb{C}^{n\times n},$
    $$\int L\left\{X\right\} = L\left\{\int X\right\}.$$
\end{lemma}
\begin{proof}
    This follows from the linearity of integration and of the Kronecker product.
\end{proof}

\begin{lemma}
\label{Commutator of Liouvillian}
    For $A(t) = -iL\left\{H(t)\right\},$
    $$\big[A(s),\,A(r)\big] = L\left\{\big[H(r),\,H(s)\big]\right\}.$$
\end{lemma}
\begin{proof}
    We first note that by the mixed-product property of the Kronecker product,
    \begin{align*}
        L\left\{H(s)\right\}\,L\left\{H(r)\right\} &= \left(I \otimes H(s) - H(s)^T \otimes I\right)\left(I \otimes H(r) - H(r)^T \otimes I\right) \\
        &= I \otimes H(s)H(r) - H(s)^T \otimes H(r) \\
        \MoveEqLeft[-5] - H(r)^T \otimes H(s) + H(s)^T H(r)^T \otimes I.
    \end{align*}
    This implies that
    \begin{align*}
        \big[-iL\left\{H(s)\right\},\,-iL\left\{H(r)\right\}\big] &= L\left\{H(r)\right\}L\left\{H(s)\right\} - L\left\{H(s)\right\}L\left\{H(r)\right\} \\
        &= I \otimes H(r)H(s) - I \otimes H(s)H(r) \\ 
        \MoveEqLeft[-2] + H(r)^T H(s)^T \otimes I - H(s)^T H(r)^T \otimes I.
    \end{align*}
    By linearity of the Kronecker product and properties of transposition,
    \begin{align*}
        \big[-iL\left\{H(s)\right\},\,-iL\left\{H(r)\right\}\big] &= I \otimes \big(H(r)H(s) - H(s)H(r)\big) \\
        \MoveEqLeft[-3] - \big(H(r)H(s) - H(s)H(r)\big)^T \otimes I \\
        &= L\left\{\big[H(r),\,H(s)\big]\right\}.
    \end{align*}
\end{proof}

\begin{theorem}
\label{Double Integral of Commutator of Liouvillian}
    $$\int_0^t \int_0^s \big[-iL\left\{H(s)\right\},\,-iL\left\{H(r)\right\}\big]\,dr\,ds = L\left\{\int_0^t \int_0^s \big[H(r),\,H(s)\big]\,dr\,ds\right\}.$$
\end{theorem}
\begin{proof}
    Apply Lemma \ref{Commutator of Liouvillian} and then apply Lemma \ref{Integral of liouvillian} twice.
\end{proof}

\subsection{Spin Components and Expectation Values}

An important aspect of a spin system is being able to discern properties and analyse how these evolve over time. From the density matrix we can determine the expectation value for a given operator. Operators are represented by $n\times n$ Hermitian matrices for system of $n$ spins. We denote the expectation value for an operator $A$ by $\langle A \rangle.$

\begin{definition}[Frobenius Inner Product]
    For $A,B\in\mathbb{C}^{n\times n}$, 
    $$\langle A,B\rangle_F := \Tr\{A^\dagger B\}.$$
\end{definition}

\begin{theorem}
For any two Hermitian matrices $A, B\in\mathbb{C}^{n\times n},$ their Frobenius inner product is real.
\end{theorem}
\begin{proof}
For a Hermitian matrix $A\in\mathbb{C}^{n\times n},$ we have $A = A^\dagger$ and can write $A^T = \overline{A}.$ Applying these and standard properties of conjugation and transposition, we get
\begin{align*}
    \Tr\{A^\dagger B\} &= \Tr\{AB\} \\ 
    &= \overline{\Tr\{AB\}} \\
    &= \overline{\Tr\{A^\dagger B\}},
\end{align*}
meaning
$$\Tr\{A^\dagger B\}\in\mathbb{R}.$$
\end{proof}
\noindent
All single spin density matrices can be expressed as a linear combination of the three Pauli matrices. When the operator is a Pauli matrix, the resulting expectation value is the component of spin in the respective direction. For example, consider the following HOCP system for time $t\in[0,5]$,

\begin{equation}
\begin{cases}
    \beta_1 = 10,\;\;\gamma_1 = 2,\;\;\Omega_1 = 1,\\
    \rho_0 = \sigma_x.
\end{cases}
\label{eqn:hocp1}
\end{equation}

\begin{figure}[H]
    \centering
    \includegraphics[width=0.9\textwidth]{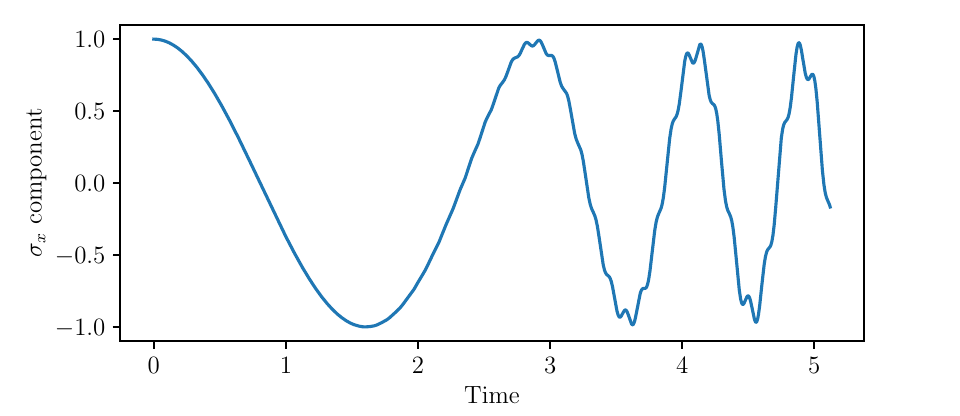}
    \caption{Normalised component of (\ref{eqn:hocp1}) in $\sigma_x$ direction over time.}
\end{figure}
\noindent
We calculate the $x$-component of spin---equivalently the expectation of the operator $\sigma_x$---using the Frobenius inner product, 
$$\langle \sigma_x \rangle(t) = \langle{\rho(t),\sigma_x\rangle}_F.$$
The normalised $x$-component is then
$$d_x(t) = \frac{1}{2}\langle{\rho(t),\sigma_x\rangle}_F.$$
And $d_y$ and $d_z$ are defined similarly. These components are such that
$$d_x(t)^2 + d_y(t)^2 + d_z(t)^2 = 1.$$
In Figure 1 at time $t=0$, we see that the $x$-component is $1$, meaning that the spin is wholly in the $x$-direction. Thus $d_y$ and $d_z$ must be zero.

For a $n$-spin system that is not interacting, with a density matrix of the form (\ref{nonIntDensityMatrix}), we can measure the $\sigma_x$ component of the $k$th spin using the operator $\mathbf{I}_k^x$. The normalisation constant is $2^n$ for $n$ particles. For example, the normalised $\sigma_x$ component of the first of two non-interacting particles is 
$$ \frac{1}{4}\langle{\rho(t),\sigma_x \otimes I\rangle}_F.$$

\subsubsection{The Bloch Sphere}

It is often difficult to visualise the spin of a quantum system, especially when there is more than one spin in the system. However, for a single spin we can use the fact that an analogy can be drawn between the normalised components and 3D space.

In particular, we can treat the three Pauli matrices as the canonical basis of $\mathbb{R}^3$ and the three components as a vector projected from the origin onto the surface of the unit sphere. This is called the Bloch sphere.

\begin{figure}[H]
    \centering
    \includegraphics[width=0.5\textwidth]{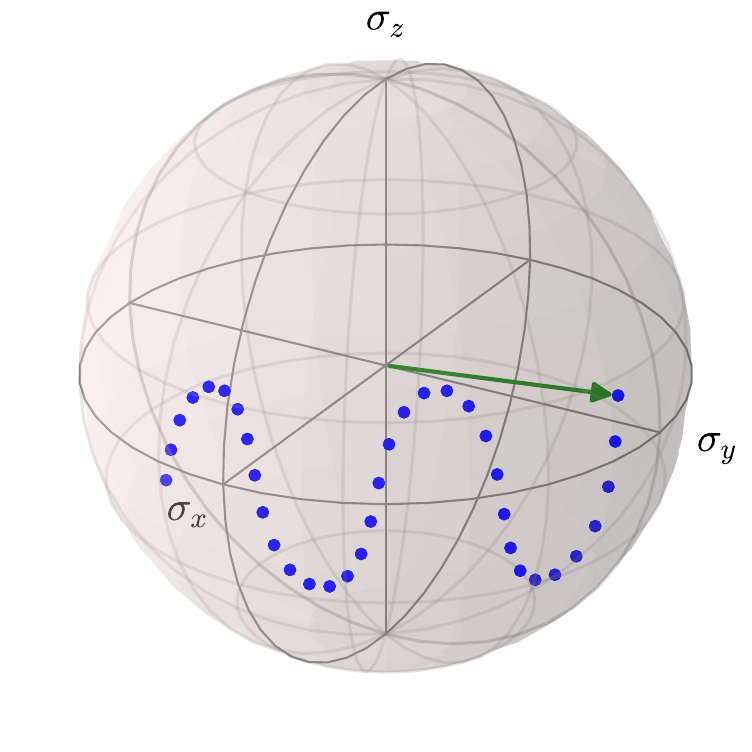}
    \caption{Evolution of (\ref{eqn:hocp1}) over a small interval visualised on the Bloch sphere.}
\end{figure}

\pagebreak

\section{Current Methods}

Here we consider different approaches to solving ordinary differential equations. We first look at some important standard methods applied to the relevant ODE and then move onto more advanced approaches using approximations of the matrix exponential, discussing their rates of convergence, error bounds, and computational efficiency.

For the initial value problem
\begin{equation}
    \mathbf{x}'(t) = A(t)\mathbf{x}(t),\;\;\;\mathbf{x}(t_0) = \mathbf{x}_0,\;\;\;t\in[t_0,\,t_f],
    \label{eqn:ode}
\end{equation}
we discretise the time interval with a time-step $h=0.5^k$, where $k\in\mathbb{N}$ is specified, and approximate the solution at each such point in time, beginning with $\mathbf{x}_0$ at $t_0$. From now on we give the step-size by specifying the value of $k$. The general solution this approach gives is
$$\mathbf{x}(t_{n+1}) = P(t_n)\,\mathbf{x}(t_n),$$
for some matrix-valued function $P(t)$. This inductively gives an explicit form for $\mathbf{x}(t_n),$
$$\mathbf{x}(t_n) = P(t_n)^n\,\mathbf{x}_0.$$
Note how equation (\ref{eqn:ode}) is analogous to the vectorised Liouville-von Neumann equation (\ref{veclvn}).

\subsection{QuTiP}
A current implementation for solving the Liouville-von Neumann equation is the open-source Python package QuTiP, which provides methods and classes for representing and manipulating quantum objects \cite{QuTiP}. The function with which we are concerned is the \texttt{mesolve} function, which evolves a density matrix using a given Hamiltonian in the Liouville-von Neumann equation, evaulating the density matrix at the given times.

As opposed to the other methods we consider, where the given times are used to approximate the density matrix, QuTiP sub-samples the density matrix at the given times from a much more accurate solution. This enables \texttt{mesolve} to maintain an accuracy of $10^{-4}$ regardless of how large a time-step is taken.

\begin{figure}[H]
    \centering
    \includegraphics[width=0.8\textwidth]{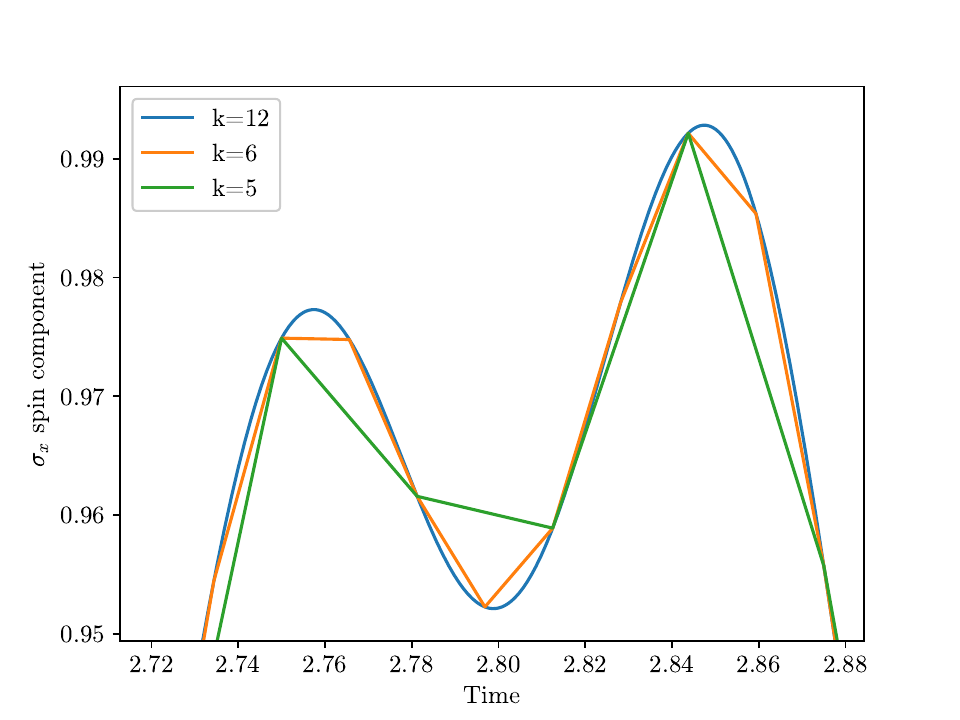}
    \caption{QuTiP sub-sampling for three solutions to (\ref{eqn:hocp1}).}
\end{figure}

\subsection{Standard ODE Methods}

\subsubsection{Euler's Method}
We begin with the most basic implicit methods for solving ordinary differential equations. There are two forms of Euler's method: explicit and implicit. 

\begin{itemize}
    \item Explicit Euler:
    $$\mathbf{x}(t_{n+1}) = \big(I + hA(t_n)\big)\,\mathbf{x}(t_n),$$
    \item Implicit Euler:
    $$\mathbf{x}(t_{n+1}) = \big(I-hA(t_{n+1})\big)^{-1}\mathbf{x}(t_n).$$
\end{itemize}
Implicit methods typically require more computations, but result in more stable solutions for larger time-steps. For Implicit Euler we see that a matrix inversion is needed, which is incredibly costly and warrants further numerical methods. 

Consider the following system:
\begin{equation}
\begin{cases}
    H(t) = \sigma_x + \sigma_y + \sigma_z, \\
    \rho_0 = \sigma_x.
\end{cases}
\label{eqn:simple_system}
\end{equation}
\noindent
Plotting the solution of (\ref{eqn:simple_system}) for these methods and QuTiP's \texttt{mesolve} for $k=5$, we see how both Euler methods violate the conservation of energy: the normalised $\sigma_x$ component oscillates between $1$ and $-1$, but these oscillations increase and decrease over time for the two methods, respectively.

\begin{figure}
    \centering
    \includegraphics[width=0.8\textwidth]{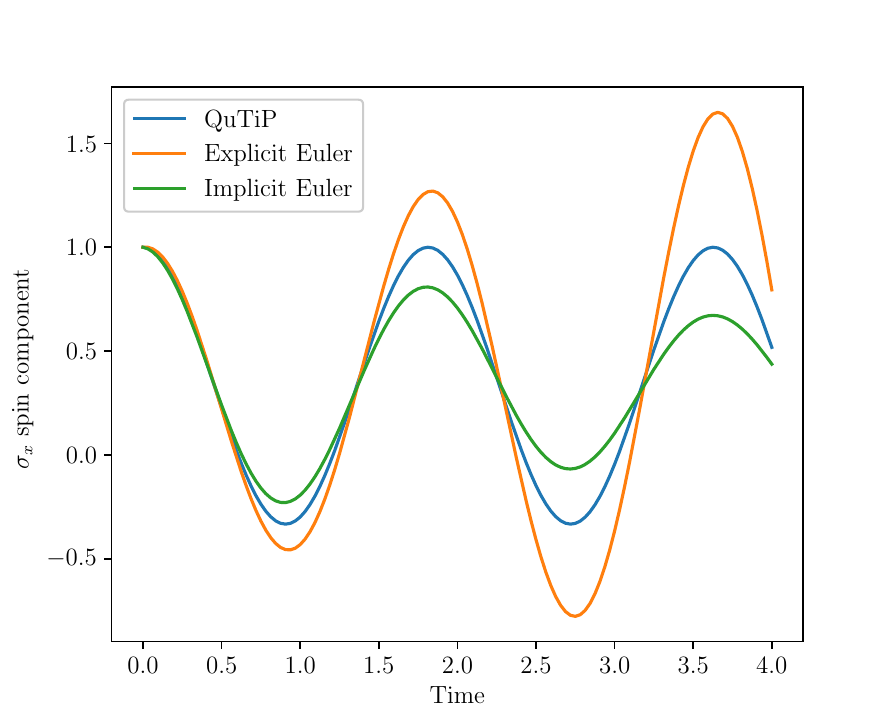}
    \caption{Normalised component of system $(2)$ in $\sigma_x$ direction over time for $k=5$.}
\end{figure}

\subsubsection{The Trapezoidal Rule}

Continuing with the same ODE introduced previously, the trapezoidal rule is as follows:
$$\mathbf{x}(t_{n+1}) = \bigg(\big(I - \frac{h}{2}A(t_n)\big)\big(I + \frac{h}{2}A(t_n)\big)\bigg)^{-1}\mathbf{x}(t_n).$$
Being a second-order implicit method, the trapezoidal rule is more stable and converges to our reference solution faster than both of the Euler methods. 

Although we won't employ this method for larger systems, it serves to highlight how midpoint methods converge faster than their initial point counterparts. To change the previous methods from initial point to midpoint, we replace $t_n$ with $t_n + h/2$.

\begin{figure}[H]
    \centering
    \includegraphics[width=0.9\textwidth]{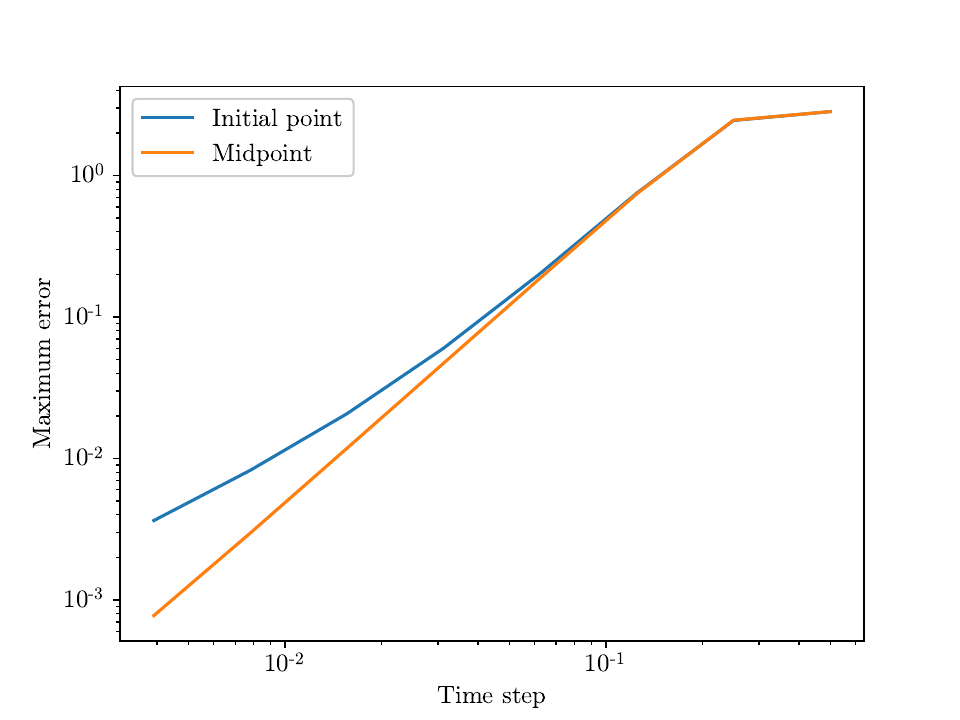}
    \caption{log-log plot of maximum error against time-step $h$ for the trapezoidal rule.}
    \label{fig:traploglog}
\end{figure}
\noindent
From Figure \ref{fig:traploglog} we see that for system $(1)$, initial point trapezoidal rule has a rate of convergence $1.255\approx1$, whereas midpoint has a rate of convergence $1.965\approx2$.

\subsubsection{RK4}
The most widely known Runge-Kutta method---commonly referred to as RK4---is an explicit fourth-order method that uses the previous value of $\mathbf{x}$ and the weighted average of four increments to determine the following value of $\mathbf{x}$. The system is evolved as follows:
$$\mathbf{x}(t_{n+1}) = \mathbf{x}(t_n) + \frac{1}{6}h(k_1 + 2k_2 + 2k_3 + k_4),$$
where 
\begin{align*}
    k_1 &= A(t_n)\,\mathbf{x}(t_n), \\
    k_2 &= A\left(t_n + \frac{h}{2}\right)\left(\mathbf{x}(t_n) + h\frac{k_1}{2}\right), \\
    k_3 &= A\left(t_n + \frac{h}{2}\right)\left(\mathbf{x}(t_n) + h\frac{k_2}{2}\right), \\
    k_4 &= A\left(t_n + h\right)\left(\mathbf{x}(t_n) + hk_3\right).
\end{align*}
When the right-hand side of the ODE is independent of $\mathbf{x}$, RK4 reduces to Simpson's rule \cite{Suli}. Interestingly, when $A(t)$ is independent of $t$, the function $P(t)$ for RK4 becomes
$$I + hA + \frac{h^2}{2!}A^2 + \frac{h^3}{3!}A^3 + \frac{h^4}{4!}A^4,$$
which is the first four terms of the Taylor expansion of the matrix exponential.

\begin{figure}[H] 
    \centering
    \includegraphics[width=0.9\textwidth]{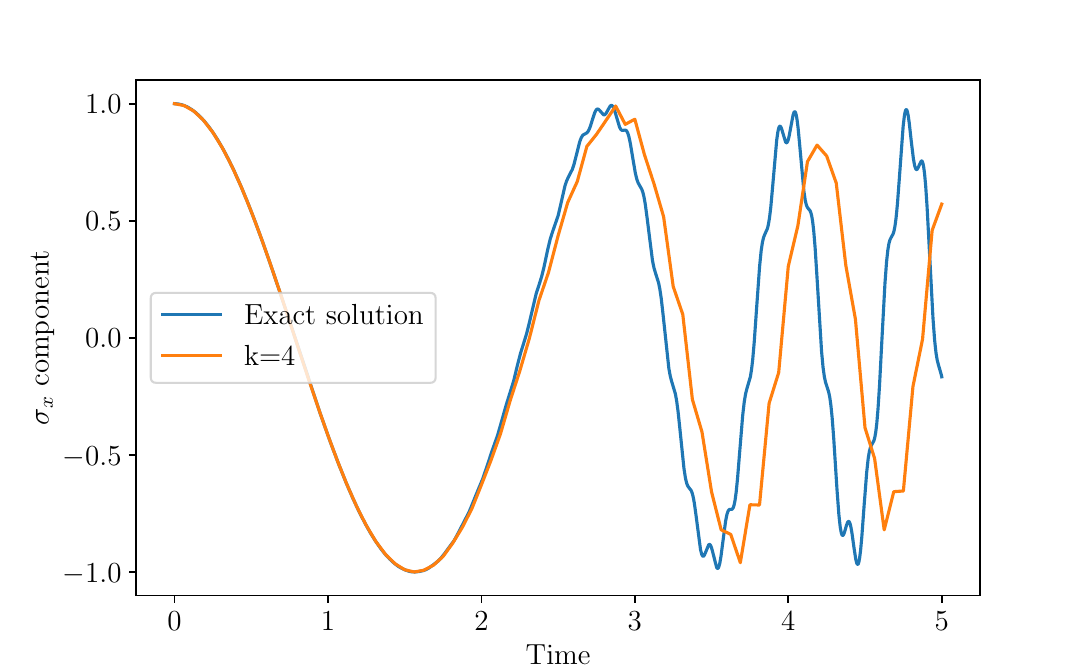}
    \caption{The true solution and the RK4 approximation of the system (\ref{eqn:hocp1}) for $k=4$.}
    \label{fig:rk4}
\end{figure}

\subsection{Approximating the Matrix Exponential}

An important function in the solution of differential equations is the matrix exponential. Here we discuss some more sophisticated methods that can be implemented in numerically solving ODEs.

\subsubsection{Taylor Series}
This method directly calculates the matrix exponential using the definition
$$e^{A} = I + A + A^2/2! + \dotsc,$$
which is truncated to the first $K$ terms, giving
$$ T_K(A) = \sum_{j=0}^K A^j/j!$$
as the approximation for $e^A$.

An arbitrarily high degree of accuracy can be achieved given a large enough $K$ is used, but even when approximating the scalar exponential this method is substandard. One may also find example matrices which cause "catastrophic cancellation" in floating point arithmetic when using "short arithmetic". This is not, however, a direct consequence of the approximation, but a result of using insufficiently high relative accuracy in the calculation \cite{Nineteen}.

Liou \cite{LiouExpm} describes an approach to calculating $e^{tA}$ via power series, writing
$$e^{tA} = M + R,$$
where $M$ is the $K$-term approximation
$$\sum_{j=0}^K \frac{t^j A^j}{j!}$$
and $R$ the remainder 
$$\sum_{j=K+1}^\infty \frac{t^j A^j}{j!}.$$
In order the achieve an accuracy of $d$ significant digits, we require 
\begin{equation}
    |r_{ij}| \leq 10^{-d}\,|m_{ij}|,
    \label{eqn:dsigfig}
\end{equation}
for the elements of $R$ and $M$. The upper bound derived by Liou for $|r_{ij}|$, given $K$ terms in the approximation, was sharpened by Everling \cite{EverlingExpm} to 
$$|r_{ij}| \leq \norm{\frac{t^K A^K}{K!}}\,\frac{t\norm{A}}{K+1}\;\frac{1}{1-\epsilon},\;\;\;\epsilon = \frac{t\norm{A}}{K+2},$$
which allows for checking whether the desired accuracy defined by (\ref{eqn:dsigfig}) has been achieved.

This approach is much worse than some of the methods addressed later, but is worth covering briefly for completeness-sake. One issue is that it does not take advantage of any of the properties of $A$, such as its eigenvalues or Hermitivity, as other methods do.

\subsubsection{Padé Approximants}
When approximating a function, the Taylor series expansion is limiting since polynomials are not a good class of functions to use if the approximated function has singularities. Rational functions are the simplest functions with singularities, making them an obvious first choice for a better approach \cite{Assche}.

The Padé approximation is a rational function with numerator of degree $\leq p$ and denominator of degree $\leq q$. We may adapt this approach to calculating the matrix exponential as follows \cite{Nineteen}:
$$ e^A \approx R_{pq}(A) = [D_{pq}(A)]^{-1} N_{pq}(A),$$
where
$$N_{pq}(A) = \sum_{j=0}^{q} \frac{(p+q-j)!p!}{(p+q)!j!(p-j)!}A^j$$
and
$$ D_{pq}(A) = \sum_{j=0}^{q} \frac{(p+q-j)!q!}{(p+q)!j!(q-j)!]}(-A)^j.$$
Rational approximations of the matrix exponential are preferred over polynomial approximations due to their better stability properties. When $p \leq q$ the Padé approximation gives rise to unconditionally stable methods \cite{VargaPade}.

When the norm of A is not too large, $p=q$ is a preferred choice for the order of the numerator and denominator. However, the computational costs and rounding errors for Padé approximants increase as $\norm{A}$ increases, where $\norm{A}$ denotes the matrix operator norm. We may mitigate this by scaling and squaring, which takes advantage of the following property:
$$e^A = (e^{A/{2^j}})^{2^j}.$$
This allows a $j$ to be chosen such that the Padé approximation is not too costly. This is the algorithm used by SciPy's \texttt{expm} function \cite{Higham}.

Given an error tolerance, we may determine an appropriate $q$ and $j$ to approximate $e^{A/{2^j}}$ using Padé approximants, and multiply to get $e^A$. For these values of $q$ and $j$, we find that Padé approximants are more efficient than a Taylor approximant. Specifically, for a small $\norm{A}$, the Taylor expansion requires about twice as much work for the same level of accuracy \cite{Nineteen}. The appendices contains a table of values for which a diagonal Padé approximant is optimal, given a norm $\norm{A}$ and error tolerance $\epsilon$.

\subsubsection{Krylov Subspace Methods}
It is not always necessary to calculate $e^A$ in its entirety. Sometimes we only need the product of the matrix exponential with a vector. Even if $A$ is a sparse matrix $e^A$ will typically still be dense, so when dealing with large matrix systems it may be computationally infeasible to calculate $e^A$ completely. This is easily seen with quantum systems, where a ten spin system results in a Liouvillian with dimension $2^{20}$. 

Given a matrix $A\in\mathbb{C}^{n\times n}$ and a compatible vector $\mathbf{b}$, we can approximate $e^A\mathbf{b}$ using a Krylov subspace method. Being iterative methods, they produce a sequence of improving approximations, and achieve the actual solution (up to arithmetic accuracy) with $n$ iterations.

We consider an order-$m$ Krylov subspace,
$$\mathcal{K}_m(A,\mathbf{v}) = \text{span}\lbrace \mathbf{v}, A\mathbf{v}, \dotsc, A^{m-1}\mathbf{v}\rbrace,$$
where $\mathbf{v} = \mathbf{b}/\norm{\mathbf{b}}$, and wish to find an orthonormal basis. The approximation to $e^A\mathbf{b}$ is an element of this space. For Hermitian matrices, we employ the Lanczos algorithm with $m\leq n$ iterations \cite{Lanczos}. The algorithm outputs an $n\times m$ matrix $V_m$ with orthonormal columns and an $m \times m$ real, symmetric, tridiagonal matrix $T_m$ such that $T_m = V_m^\dagger\,A\,V_m.$ For $m=n$, $V_m$ is a unitary matrix and so $A \approx V_m\,T_m\,V_m^\dagger.$ Then,
$$e^A\,\mathbf{b} \,\approx\, \norm{\mathbf{b}} V_m\,e^{T_m}\,e_1,$$
where $e_1$ is the first column of the $m$-dimensional identity matrix. Since $T_m$ is tridiagonal, Padé approximants can be used to calculate $e^{T_m}$ efficiently \cite{SaadPadeKrylov}.

Krylov methods are preferred over polynomial or rational function approximations for large and/or sparse matrices, since it replaces calculating $e^A$ with calculating the much easier $e^{T_m}$. This is achieved through approximating the eigenvalues of $A$ with the eigenvalues of $H_m$.

For a negative semi-definite Hermitian matrix $A$, with eigenvalues in the interval $\left[-4\rho,0\right]$, a scalar $t$, and a unit vector $v$, Hochbruck and Lubich \cite{HochbruckKrylov} derived a bound for the error,
$$\varepsilon_m := \norm{e^{tA}v - V_me^{tT_m}e_1}.$$
The error is bounded as follows:
\begin{equation}
\varepsilon_m \leq
\begin{cases}
    10\,e^{-m^2 / (5\rho t)},\;\;\;\;\;\;\;\;\;\;\;\;\;\;\;\;\;\;\;\sqrt{4\rho t\vphantom{|}}\leq m \leq 2\rho t, \\
    10(\rho t)^{-1} e^{-\rho t}\Big(\dfrac{e\rho t}{m}\Big)^m,\;\;\;\;\;\;\;\;\;\;\;\;\;\;\;\;\;\;\;m\geq 2\rho t.
\end{cases}
\label{eqn:KrylovError}
\end{equation}

\begin{figure}[H]
    \centering
    \includegraphics[width=\textwidth]{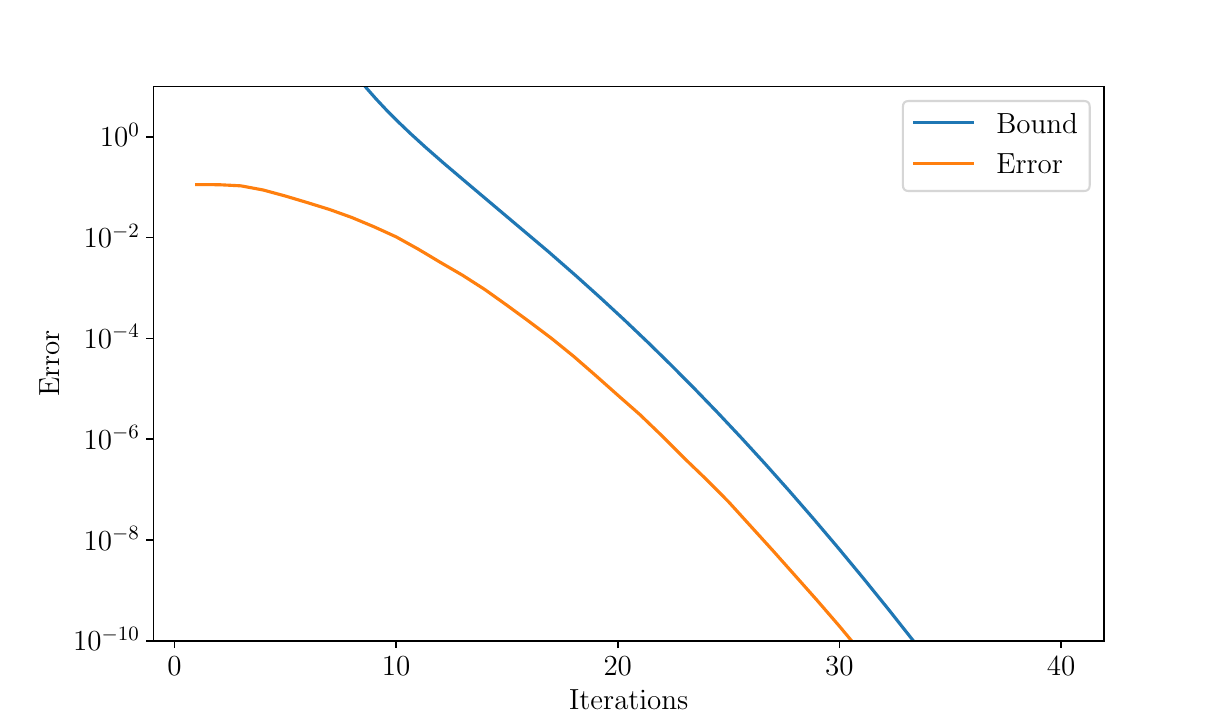}
    \caption{The error bound (\ref{eqn:KrylovError}) and the experimental error of our example.}
    \label{fig:KrylovError}
\end{figure}
\noindent
Following the approach of \cite{HochbruckKrylov}, we consider a diagonal matrix with eigenvalues in the interval $[-40,0]$ and a 1001-dimensional unit vector. This error bound and the error in the approximation from our implementation---see appendix---is shown in Figure \ref{fig:KrylovError}.

\pagebreak

\section{The Magnus Expansion}
In this section we aim to implement a method which converges faster than current approaches by applying the Magnus expansion to the differential equation (\ref{eqn:ode}), providing an exponential representation of the solution to such an equation \cite{Magnus}. Solutions derived from the Magnus expansion give approximations which preserve certain qualitative properties of the exact solution, making it a desirable method to consider \cite{BlanesMagnus}.

We first discuss the theory behind the approach, detailing the components relevant to our implementation, and then optimise the solution to best solve ODEs driven by Hamiltonians in the form of (\ref{eqn:Hamiltonian}). In addition to this, we describe how our code was implemented in the Python package MagPy and give an overview of how to use it.

\subsection{Definition}
Given an $n\times n$ time-dependent matrix $A(t)$, the Magnus expansion provides a solution to the initial value problem
$$\mathbf{x}'(t) = A(t)\,\mathbf{x}(t),\;\;\mathbf{x}(t_0)=\mathbf{x}_0.$$
for a vector $\mathbf{x}(t)$. The solution takes the form
$$\mathbf{x}(t) = \exp{\left(\Omega(t,t_0)\right)}\,\mathbf{x_0},$$
where $\Omega(t, t_0)$ is the infinite series,
$$ \Omega(t, t_0) = \sum_{k=1}^{\infty}\Omega_k(t),$$
where $\Omega_k(t)$ are matrix-valued functions. We will apply this approximation iteratively, using the same approach as described at the beginning of section 4. For $t_0 = 0$ we write $\Omega(t,0)$ as $\Omega(t)$.
Arbitrarily high degrees of accuracy can be achieved given a sufficient number of terms in the summation, but we will truncate to the first two terms,
\begin{align}
    \Omega_1(t) &= \int_0^t A(s)\,ds, \\
    \Omega_2(t) &= \int_0^t \int_0^s \left[A(s),A(r)\right]dr\,ds.
    \label{eqn:MagnusTerms}
\end{align}
\noindent
When $A(s)$ and $A(r)$ commute $\Omega(t)$ reduces to $\Omega_1(t)$, and for a time independent $A$ the commutator terms cancel and this reduces to the well-known solution
$$\mathbf{x}(t) = \exp\{\left(t-t_0)A\right)\}\mathbf{x_0}.$$

\subsection{Reformulation}
Since the Liouvilian takes the form of a matrix-valued function, it is difficult to implement the integral terms in the Magnus expansion directly. The core idea here is to write the integrals in terms of scalar integrals of $f_j(t)$, $g_j(t)$, and $\Omega_j$, reducing the complexity of the code and the number of computations needed. We consider the two terms separately, detailing the necessary properties used in implementing the Magnus method in Python. For the first term the following theorem is sufficient to write the integral in the required form.

\begin{theorem}[1st Term of Magnus Expansion]
\label{thm:Magnus1stTerm}
    For $A(t) = -iL\left\{H(t)\right\},$
    $$\int_0^t A(s)\,ds = -iL\left\{\sum_{j=1}^n K_j(s) + H_J\,t\right\},$$
    where 
    $$K_j(s) = \int_0^t f_j(s)\,ds\,\mathbf{I}_j^x + \int_0^t g_j(s)\,ds\,\mathbf{I}_j^y + \Omega_j\,t\,\mathbf{I}_j^z.$$
\end{theorem}
\begin{proof}
    Lemma \ref{Integral of liouvillian} gives us
    $$\int_0^t A(s)\,ds = -iL\left\{\int_0^t H(s)\,ds\right\},$$
    and applying Lemma \ref{Integral of Hamiltonian} to $\int_0^t H(s)\,ds$ gives the result.
\end{proof}

\subsubsection{Derivation of the Second Term}
For the second term of the expansion we are concerned with calculating a double integral of the commutator,
$$\big[A(s),A(r)\big],$$
where 
$$A(t) = -iL\left\{H(t)\right\}.$$
We build up to this through the following results. Theorem \ref{Double Integral of Commutator of Liouvillian} reduces this to first calculating the commutator of 
$$H(t) = \sum_{j=1}^n \mathbf{I}_j\left\{H_j(t)\right\} + H_J,$$
integrating, and then applying the Liouvillian superoperator. 
\begin{lemma}
\label{I of Single Hamiltonian Commutator}
    For $H_j(t) = f_j(t)\,\sigma_x + g_j(t)\,\sigma_y + \Omega_j\,\sigma_z,\;j\in\mathbb{N},$
    \begin{multline*}
        \mathbf{I}_j\left\{\big[H_j(r),\,H_j(s)\big]\right\} = 2\,i\,\Omega_j\big(g_j(r) - g_j(s)\big)\mathbf{I}_j^x \\
        + 2\,i\,\Omega_j\big(f_j(s) - f_j(r)\big)\mathbf{I}_j^y \\ 
        + 2\,i\,\big(f_j(r)g_j(s) - g_j(r)f_j(s)\big)\mathbf{I}_j^z.
    \end{multline*}
\end{lemma}
\begin{proof}
    Expanding $\big[H_j(r),\,H_j(s)\big]$, applying the commutator properties of the Pauli matrices, and simplifying gives
    \begin{multline*}
        \big[H_j(r),\,H_j(s)\big] = 2\,i\,\Omega_j\big(g_j(r) - g_j(s)\big)\sigma_x + 2\,i\,\Omega_j\big(f_j(s) - f_j(r)\big)\sigma_y \\ + 2\,i\,\big(f_j(r)g_j(s) - g_j(r)f_j(s)\big)\sigma_z.
    \end{multline*}
    Then the linearity of $\mathbf{I}$ yields the result.
\end{proof}

\begin{corollary}
\label{Integral of I of Single Hamiltonian Commutator}
    \begin{align*}
        \int_0^t \int_0^s \mathbf{I}_j\left\{\big[H_j(r),\,H_j(s)\big]\right\}dr\,ds \\ 
        = 2\,i&\,\Omega_j\int_0^t\int_0^s g_j(r) - g_j(s)\,dr\,ds\,\mathbf{I}_j^x \\ 
        +& 2\,i\,\Omega_j\int_0^t\int_0^s f_j(s) - f_j(r)\,dr\,ds\,\mathbf{I}_j^y \\ 
        +& 2\,i\int_0^t\int_0^s f_j(r)g_j(s) - g_j(r)f_j(s)\,dr\,ds\,\mathbf{I}_j^z.
    \end{align*}
\end{corollary}
\begin{proof}
    This follows directly from Lemma \ref{I of Single Hamiltonian Commutator} and from linearity of integration.
\end{proof}
\noindent
Now we look at the commutator of $H_0(t) = \sum_{j=1}^n \mathbf{I}_j\left\{H_j(t)\right\}$ and prove a useful result 
\begin{lemma}
\label{Commutator of non-interacting Hamiltonian}
    Given $H_0(t) = \sum_{j=1}^n \mathbf{I}_j\left\{H_j(t)\right\},$
    $$\big[H_0(r),\,H_0(s)\big] = \sum_{j=1}^n \mathbf{I}_j\left\{\left[H_j(r),H_j(s)\right]\right\}.$$
\end{lemma}
\begin{proof}
    By linearity of summation and Lemma \ref{S properties}.1,
    \begin{align*}
        \big[H_0(r),\,H_0(s)\big] &= \sum_{i,\,j=1}^n \left(\mathbf{I}_i\big\{H_i(r)\big\}\,\mathbf{I}_j\big\{H_j(s)\big\} - \mathbf{I}_j\big\{H_j(s)\big\}\,\mathbf{I}_i\big\{H_i(r)\big\}\right), \\
        &= \sum_{i,\,j=1}^n \left(\mathbf{S}_{i,\,j}\big\{H_i(r),H_j(s)\big\} - \mathbf{S}_{j,\,i}\big\{H_j(s),H_i(r)\big\}\right).
    \end{align*}
    Applying Lemma \ref{S properties}.2,
    $$ \mathbf{S}_{i,\,j}\big\{H_i(r),H_j(s)\big\} = \mathbf{S}_{j,\,i}\big\{H_j(s),H_i(r)\big\},$$
    for $i\neq j$. And so, by Definition \ref{S},
    \begin{align*}
        \big[H_0(r),\,H_0(s)\big] &= \sum_{j=1}^n \left(\mathbf{S}_{j,\,j}\big\{H_j(r),H_j(s)\big\} - \mathbf{S}_{j,\,j}\big\{H_j(s),H_j(r)\big\}\right) \\
        &= \sum_{j=1}^n \left(\mathbf{I}_j\big\{H_j(r)\,H_j(s)\big\} - \mathbf{I}_j\big\{H_j(s)\,H_j(r)\big\}\right) \\
        &= \sum_{j=1}^n \mathbf{I}_j\left\{\big[H_j(r),\,H_j(s)\big]\right\}.
    \end{align*}
\end{proof}

\begin{corollary}
\label{Integral of non-interacting commutator}
    \begin{multline*}
        \int_0^t \int_0^s \big[H_0(r),\,H_0(s)\big]\,dr\,ds \\ = \sum_{j=1}^n\bigg( 2\,i\,\Omega_j\int_0^t\int_0^s g_j(r) - g_j(s)\,dr\,ds\,\mathbf{I}_j^x \\ 
        + 2\,i\,\Omega_j\int_0^t\int_0^s f_j(s) - f_j(r)\,dr\,ds\,\mathbf{I}_j^y \\ 
        + 2\,i\int_0^t\int_0^s f_j(r)g_j(s) - g_j(r)f_j(s)\,dr\,ds\,\mathbf{I}_j^z\bigg).
    \end{multline*}
\end{corollary}
\begin{proof}
    Integrate Lemma \ref{Commutator of non-interacting Hamiltonian} and apply Corollary \ref{Integral of I of Single Hamiltonian Commutator}.
\end{proof}
\noindent
Next, we consider the commutator of $H(t) = H_0(t) + H_J.$ By expanding and simplifying we can write
$$\big[H(r),H(s)\big] = \big[H_0(r),H_0(s)\big] + \big[H_J,H_0(s)\big] + \big[H_0(r),H_J\big].$$
We have already dealt with the first term on the right hand side, and now deal with the other two terms. The following lemma and corollary tell us how.
\begin{lemma}
\label{Commutator of Interacting Part}
    Given $H_0(t) = \sum_{j=1}^n \mathbf{I}_j\left\{H_j(t)\right\}$ and $H_J\in\mathbb{C}^{n^2 \times n^2},$
    \begin{multline*}
        \big[H_J,H_0(s)\big] + \big[H_0(r),H_J\big] = \sum_{j=1}^n \bigg(\left(f_j(r)-f_j(s)\right)\big[\mathbf{I}_j^x,H_J\big] \\ + \left(g_j(r)-g_j(s)\right)\big[\mathbf{I}_j^y,H_J\big]\bigg).
    \end{multline*}
\end{lemma}
\begin{proof}
    Note that
    \begin{align*}
        \big[H_J,H_0(s)\big] &= \sum_{j=1}^n \bigg(H_J \big(f_j(s)\mathbf{I}_j^x + g_j(s)\mathbf{I}_j^y + \Omega_j\mathbf{I}_j^z\big) \\ 
        \MoveEqLeft[-6] - \big(f_j(s)\mathbf{I}_j^x + g_j(s)\mathbf{I}_j^y + \Omega_j\mathbf{I}_j^z\big) H_J \bigg), \\
        &= \sum_{j=1}^n f_j(s)\big[H_J,\mathbf{I}_j^x\big] + g_j(s)\big[H_J,\mathbf{I}_j^y\big] + \Omega_j\big[H_J,\mathbf{I}_j^z\big],
    \end{align*}
    and similarly that
    \begin{align*}
        \big[H_0(r),H_J\big] &= \sum_{j=1}^n f_j(r)\big[\mathbf{I}_j^x,H_J\big] + g_j(r)\big[\mathbf{I}_j^y,H_J\big] + \Omega_j\big[\mathbf{I}_j^z,H_J\big].
    \end{align*}
    Summing these two equations gives the result.
\end{proof}

\begin{corollary}
\label{Integral of interacting part}
    \begin{multline*}
        \int_0^t \int_0^s \big[H_J,H_0(s)\big] + \big[H_0(r),H_J\big]\,dr\,ds\\ = \sum_{j=1}^n \bigg(\int_0^t \int_0^s f_j(r)-f_j(s)\,dr\,ds\big[\mathbf{I}_j^x,H_J\big] \\ + \int_0^t \int_0^s g_j(r)-g_j(s)\,dr\,ds\big[\mathbf{I}_j^y,H_J\big]\bigg).
    \end{multline*}
\end{corollary}
\begin{proof}
    This follows directly from Lemma \ref{Commutator of Interacting Part} and linearity of integration.
\end{proof}
\noindent
We now have all we need to write the main result for the second term,
$$\Omega_2(t) = \frac{1}{2}\int_0^t \int_0^s \left[A(s),A(r)\right]dr\,ds.$$

\begin{theorem}[2nd Term of Magnus Expansion]
\label{thm:Magnus2ndTerm}
    Given 
    $$H(t) = \sum_{j=1}^n f_j(t)\mathbf{I}_j^x + g_j(t)\mathbf{I}_j^y + \Omega_j\mathbf{I}_j^z + H_J,$$
    we can rewrite 
    $$\Omega_2(t) = \frac{1}{2}\int_0^t \int_0^s \left[A(s),A(r)\right]dr\,ds$$ 
    as
    \begin{multline*}
        \Omega_2(t) = \frac{1}{2} L\Bigg\{\sum_{j=1}^n \bigg( \int_0^t\int_0^s g_j(r) - g_j(s)\,dr\,ds\,\big(2\,i\,\Omega_j\,\mathbf{I}_j^x + \big[\mathbf{I}_j^y,H_J\big]\big) \\ 
        - \int_0^t\int_0^s f_j(r) - f_j(s)\,dr\,ds\,\big(2\,i\,\Omega_j\mathbf{I}_j^y + \big[H_J,\mathbf{I}_j^x\big]\big) \\ 
        + 2\,i\int_0^t\int_0^s f_j(r)g_j(s) - g_j(r)f_j(s)\,dr\,ds\,\mathbf{I}_j^z\bigg)\Bigg\}.
    \end{multline*}
\end{theorem}
\begin{proof}
    We begin with
    $$\Omega_2(t) = \frac{1}{2} \int_0^t \int_0^s \left[A(s),A(r)\right]dr\,ds.$$
    Theorem \ref{Double Integral of Commutator of Liouvillian} gives
    $$\frac{1}{2} L \big\{\int_0^t \int_0^s \left[H(r),H(s)\right]dr\,ds\big\}.$$
    Writing 
    $$\left[H(r),H(s)\right] = \left[H_0(r),H_0(s)\right] + \left[H_J,H_0(s)\right] + \left[H_0(r),H_J\right],$$
    we apply Corollary \ref{Integral of non-interacting commutator} and Corollary \ref{Integral of interacting part} and simplify to get the result.
\end{proof}

\subsection{Implementation in Python}

The results of Theorems \ref{thm:Magnus1stTerm} and \ref{thm:Magnus2ndTerm} are applied in the MagPy function \texttt{lvnsolve}. Given a Hamiltonian of the form (\ref{eqn:Hamiltonian}), an initial density matrix $\rho_0$, and a list of times, this function evolves a density matrix under the Liouville-von Neumann equation and returns the density matrix evaluated across the specified times, starting at $\rho_0$. Currently, the integrals in the Magnus expansion are calculated using SciPy's \texttt{integrate} package.

We introduce the function \texttt{linspace} to construct the list of times. This is a simple wrapper of the NumPy \texttt{linspace} function, using a given step-size instead of a specified number of points. \texttt{MagPy} also has built-in definitions of the Pauli matrices, which are $2\times 2$ NumPy \texttt{ndarrays}: $\sigma_x$ is denoted \texttt{sigmax} and the rest are defined analogously. The following code snippet shows how to apply \texttt{lvnsolve} to a two spin time-dependent Hamiltonian.

\begin{listing}[H]
\begin{minted}{python}
import numpy as np
import magpy as mp

def f1(t): return t
def g1(t): return np.sqrt(t)
omega1 = 1
def f2(t): return -t**2
def g2(t): return np.sin(t)
omega2 = 2

H_coeffs = [[f1, g1, omega1], [f2, g2, omega2]]
HJ = np.kron(mp.sigmax(), mp.sigmay())
rho0 = np.kron(mp.sigmax(), mp.sigmax())
tlist = mp.linspace(0, 20, 0.5**10)

density_matrices = mp.lvnsolve(H_coeffs, rho0, tlist, HJ)
\end{minted}
\caption{Example of the \texttt{lvnsolve} function.}
\end{listing}

\pagebreak

\section{Numerical Comparison}

In this section we compare the rate of convergence for different variations on Magnus-based techniques for simulating spin systems displaying HOCP. As a reference solution we use one-term Magnus with the integrals approximated using the midpoint method:
$$\int_a^b f(t)\,dt = (b-a)\,f\left(\frac{a+b}{2}\right).$$
\noindent
Our interval of simulation is $[0, 20]$ unless otherwise specified, with a time-step defined by $k=20$. We present our analysis using log-log plots for the maximum error of our approximation from the chosen reference against the time-step.

\subsection{One-term Magnus}
Recall for a single spin density matrix driven by the Liouville-von Neumann equation with a Hamiltonian of the form (\ref{eqn:Hamiltonian}), one-term Magnus approximates the density matrix as follows:
$$\mathbf{x}(t_{n+1}) = \exp{\left\{\Omega(t_{n+1},t_n)\right\}}\,\mathbf{x}(t_n),$$
where
$$\Omega(t_{n+1},t_n) = \int_{t_n}^{t_{n+1}}f_1(s)\,ds\,\sigma_x + \int_{t_n}^{t_{n+1}}g_1(s)\,ds\,\sigma_y + \Omega_1(t_{n+1} - t_n)\,\sigma_z.$$

\begin{figure}[H]
    \centering
    \includegraphics[width=\textwidth]{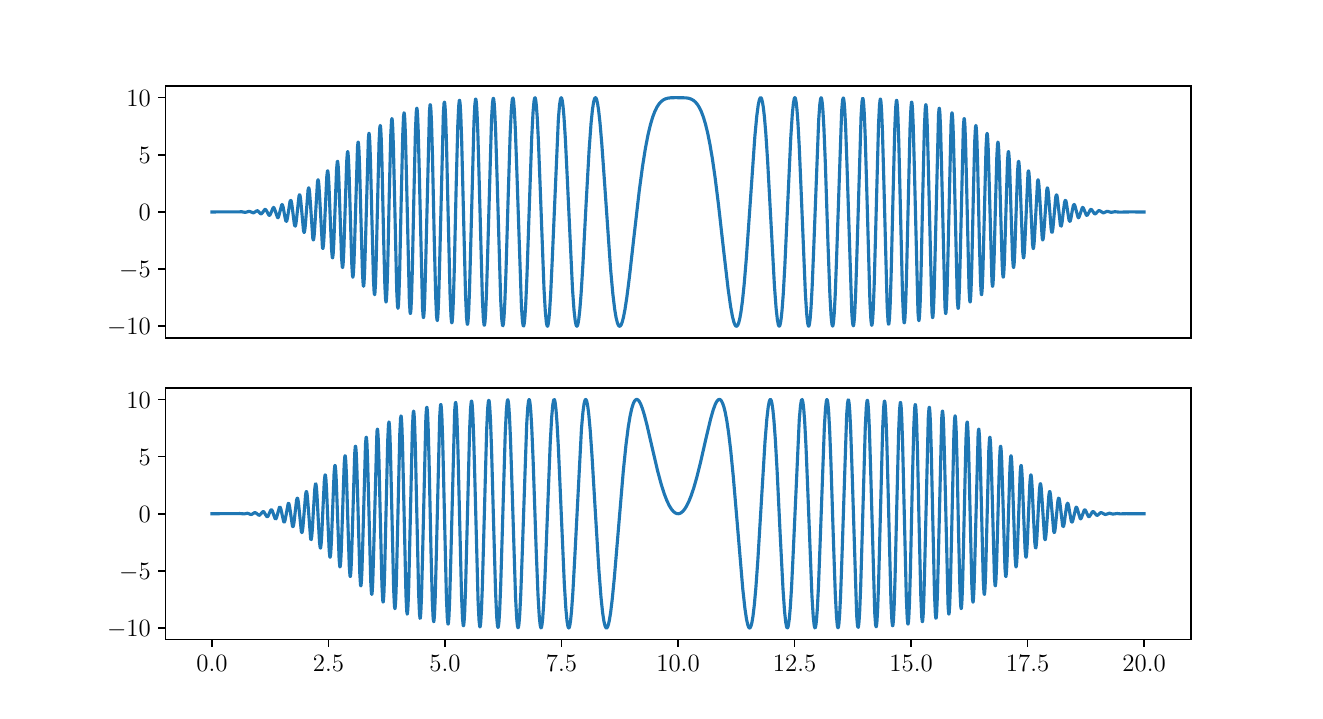}
    \caption{Functions $f_1(t)$ and $g_1(t)$ for (\ref{eqn:hocp_comparison}).}
    \label{a1_f_g}
\end{figure}
\noindent
We apply our analysis of one-term Magnus to the HOCP system defined as follows:
\begin{equation}
\begin{cases}
    \beta_1 = 10,\;\;\gamma_1 = 2,\;\;\Omega_1 = 1,\\
    \rho_0 = \sigma_x.
\end{cases}
\label{eqn:hocp_comparison}
\end{equation}
The integral approximations that we consider are initial point, midpoint, and SciPy's \texttt{integrate.quad}, which uses a Clenshaw-Curtis method with Chebyshev moments \cite{SciPy}. SciPy's method will be treated as the exact integral. We compare the error in these approximations for (\ref{eqn:hocp_comparison}) with $\beta_1$ and 
$\Omega_1$ scaled by 1, 2, $\frac{1}{10}$, and $\frac{1}{100}$. Physically this is equivalent to changing the Hamiltonian's field strength.

We can clearly see in Figure \ref{hocp1_4} the advantage of using a midpoint approximation over initial point. The midpoint approximation maintains second-order convergence, whereas initial point displays approximately first-order convergence for all simulations.

\begin{figure}[H]
    \centering
    \includegraphics[width=\textwidth]{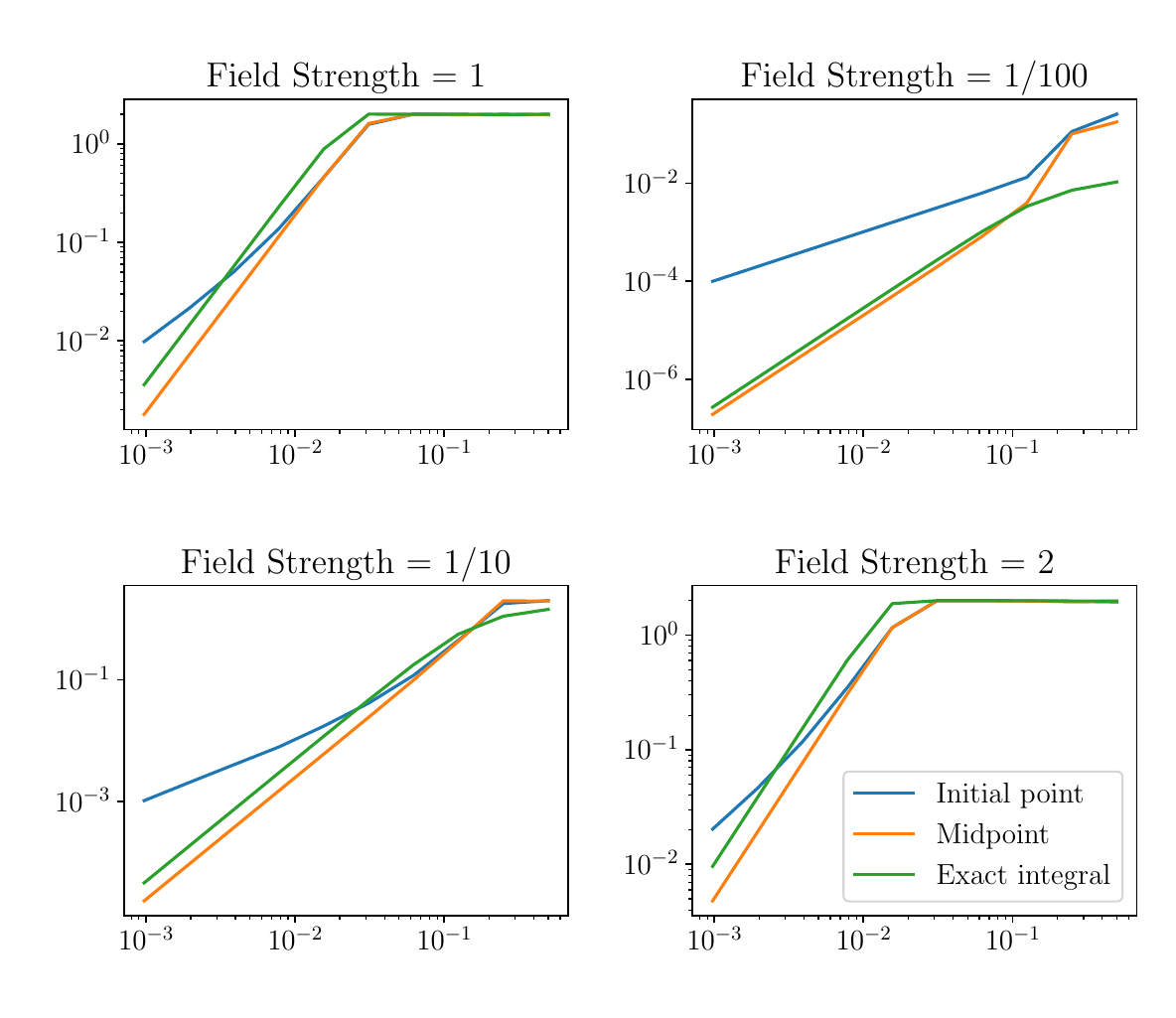}
    \caption{Error against time-step for (\ref{eqn:hocp_comparison}) with the Hamiltonian's field strength scaled four different ways.} 
    \label{hocp1_4}
\end{figure}
\noindent
Interestingly, midpoint approximations for the integral display second order convergence, which is equivalent to the rate of convergence when using the exact integral. However, midpoint begins to converge at a slightly larger time-step than when using the exact integral. For scales $\frac{1}{10}$ and $\frac{1}{100}$ we see that the exact integral method initially converges faster, but is overtaken by midpoint at a time-step of size $10^{-1}$.

It must be emphasised how surprising these results are. We see that it is not necessary to use more accurate approximation techniques for the integrals when the field strength is stronger. Plus, the midpoint approximations actually perform better than the exact integral. The complete opposite behaviour is expected. 

We also see that the smaller the scale of the system, the smaller the error is in the approximation, with a maximum error of $10^{-6}$ being achieved when the system is scaled by $\frac{1}{100}.$

\subsection{Two-term Magnus}

For two-term of Magnus, there are multiple ways to approximate the integrals in the calculation. The double integrals will all be approximated using SciPy's \texttt{integrate.quad} function being treated as the exact integral, and the single integrals will be approximated using the following:
\begin{enumerate}
    \item SciPy's \texttt{integrate.quad} function,
    \item Midpoint,
    \item Gauss-Legendre quadrature of order 3.
\end{enumerate}
Gauss-Legendre quadrature approximates integrals using quadrature weights and roots of Legendre polynomials \cite{JacobiGLQ}. Order 3 gives the following approximation:
$$\int_a^b f(x)\,dx \approx \frac{b-a}{2}\sum_{i=1}^3 w_i\,f\left(\frac{b-a}{2}\xi_i + \frac{a+b}{2}\right),$$
where
\begin{align*}
    \xi_1 &= 0,\;\;\;\;\;\;\;\;\;\,w_1 = \frac{8}{9},\\
    \xi_2 &= \sqrt{\frac{3}{5}},\;\;\;\;\;\,w_2 = \frac{5}{9},\\
    \xi_3 &= -\sqrt{\frac{3}{5}},\;\;\;w_3 = \frac{5}{9}.\\
\end{align*}
We compare the convergence of these methods against one-term Magnus using midpoint integral approximations, using midpoint as a reference solution. System (\ref{eqn:hocp_comparison})

In figure \ref{hocp1_4} we see that both two-term Magnus with SciPy and with order 3 Gauss-Legendre Quadrature display fourth-order convergence. When the system is scaled by $\frac{1}{10}$ and $\frac{1}{100}$, we achieve the accuracy of one-term Magnus with midpoint with a larger time-step of $10^{-2}$. The log-log plots plateau at a maximum error of $10^{-9}$.

Up to a time-step of $10^{-1}$ the exact integral method converges faster than Gauss-Legendre quadrature, but then they display an equal rate of convergence from that point onward. This behaviour is not seen for the system scaled by 1 or 2.

Using the midpoint method to approximate the first term in the two-term Magnus expansion is not any better than one-term Magnus using a midpoint approximation. 

\begin{figure}[H]
    \centering
    \includegraphics[width=\textwidth]{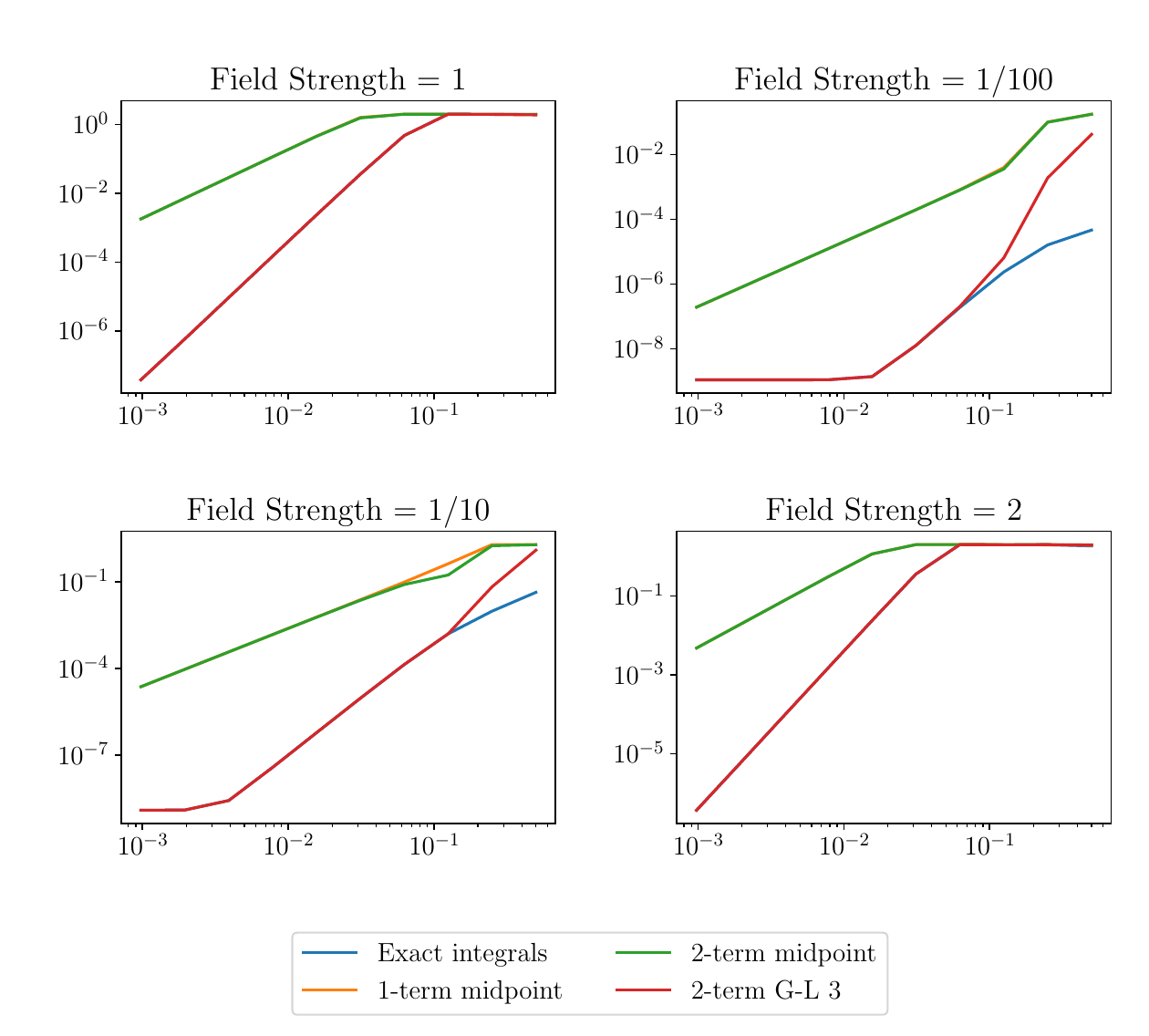}
    \caption{Error against time-step for (\ref{eqn:hocp_comparison}) with the Hamiltonian's field strength scaled four different ways.}
    \label{a1_hocp_mag2}
\end{figure}

\pagebreak

\subsection{Larger Systems}

Consider the following system of two spins:

\bgroup
\def\arraystretch{1.7}
\begin{table}[H]
\centering
\begin{tabular}{|c|c|c|c|}
    \hline
    \diagbox[innerwidth=1.5cm]{$i$}{}
    & $\beta_i$ & $\gamma_i$ & $\Omega_i$ \\
    \hline
    1 & 10 & 2 & 5 \\
    \hline
    2 & -40 & 25 & -12 \\
    \hline
\end{tabular}
\end{table}
\egroup
\noindent
The interacting component and initial condition are
$$H_J = \sigma_x \otimes \sigma_y, \;\;\;\;\;\rho_0 = \sigma_x \otimes I + I \otimes \sigma_y.$$
\noindent
For the non-interacting case, we can measure the z-component of the first spin without any interference from the second spin. 

\begin{figure}[H]
    \centering
    \includegraphics[width=\textwidth]{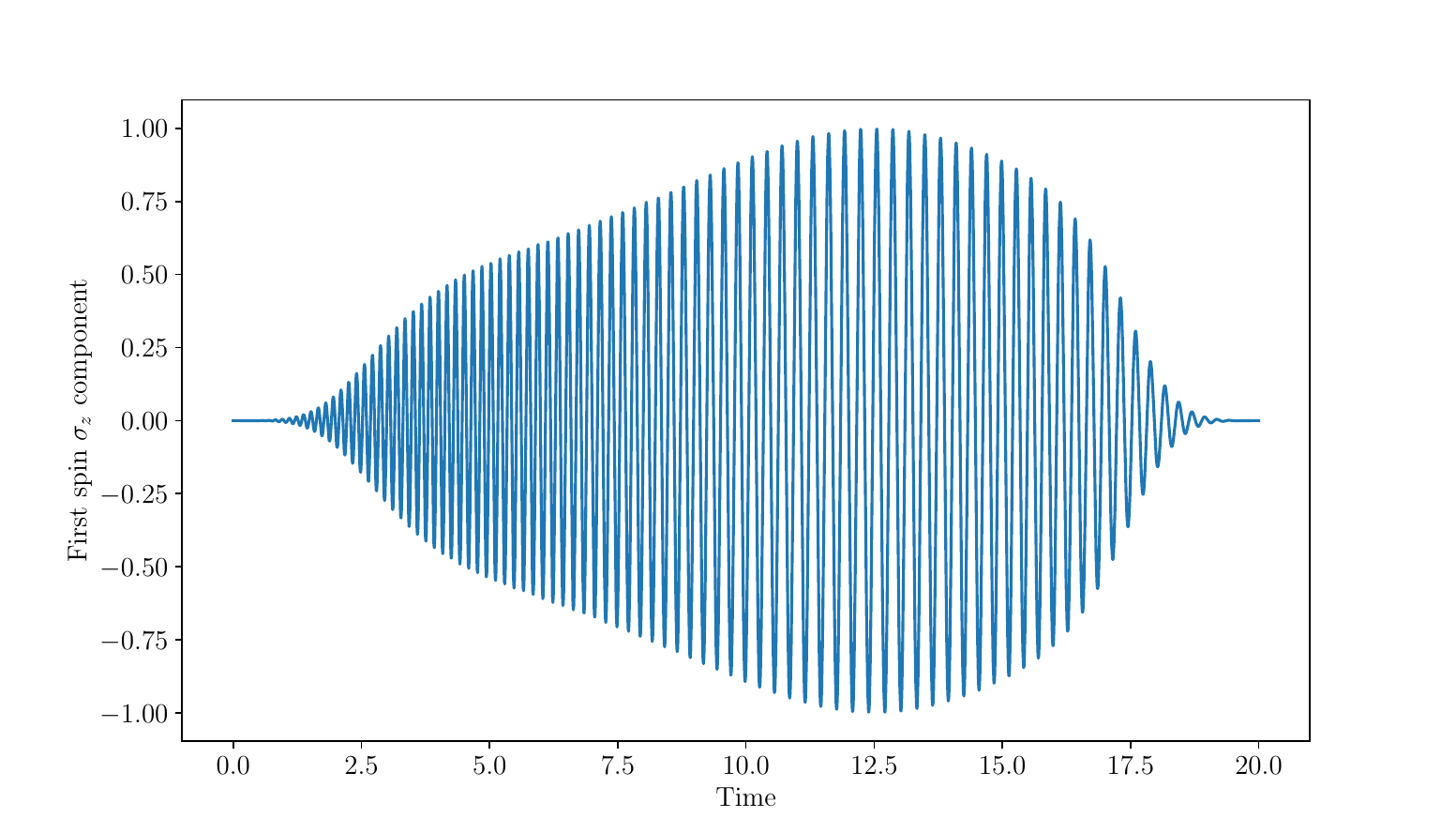}
    \caption{The normalised z-component of the first spin of system of the two-spin system.}
\end{figure}
\noindent
It is not possible to isolate a single spin in an interacting multi-spin system, but we can use a compound operator to visualise the evolution of the system. Specifically, we measure the proportion of the total density matrix in the direction $\sigma_x\otimes\sigma_x.$ Figure \ref{compound} highlights the highly-oscillatory nature of the system, showing how an insufficiently small time-step will cause these oscillations to be missed in the simulation.

\begin{figure}[H]
    \centering
    \includegraphics[width=\textwidth]{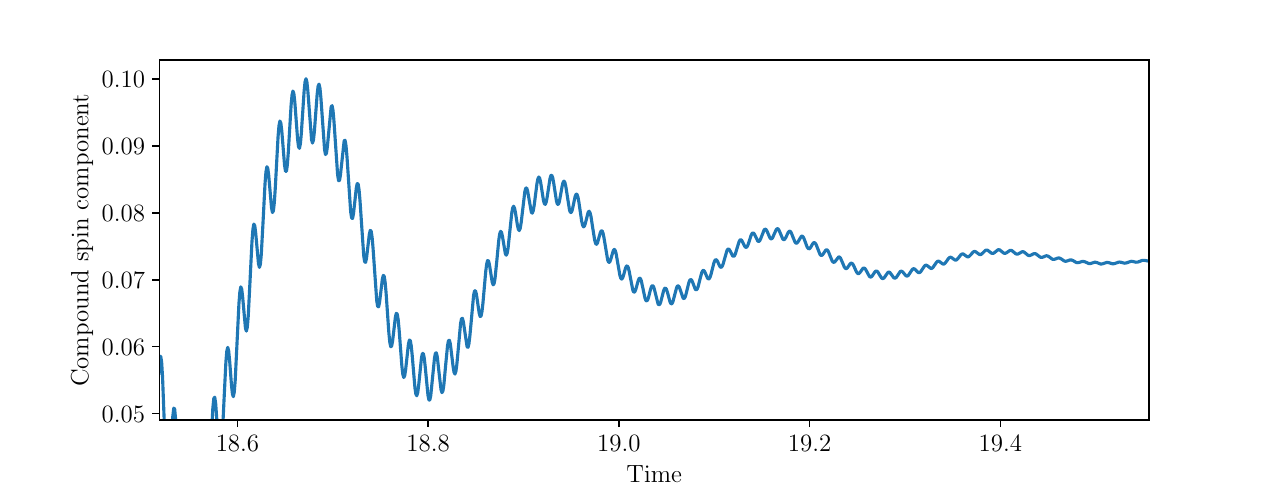}
    \caption{A snippet of the evolution of the compound spin direction of the two particle system.}
    \label{compound}
\end{figure}

\pagebreak

\section{Discussion}

\subsection{Summary}

In this project we have implemented a more efficient method for simulation of quantum spin systems displaying highly oscillatory chirped pulses using a truncated Magnus expansion. The approach has been optimised to handle a specific Hamiltonian common in the work of NMR and MRI. We opted for the Liouvillian picture in which to evolve our spin systems, and thus proved properties of the Liouvillian superoperator to write the matrix-valued integrals of the Magnus expansion in terms of scalar-valued integrals. 

Using the midpoint integral approximation as a reference, we compared the effectiveness of current numerical methods and approximations of the matrix exponential, looking at their rates of convergence and error bounds. This was then compared against our implementation of the Magnus expansion, resulting in us achieve fourth-order convergence---twice as fast as current implementations. We looked at the same system scaled four different ways and found that for systems with smaller oscillations our implementation achieves higher accuracy. We also discussed the Python package QuTiP and its approach to simulating the same spin systems.

This has been collated---along with other relevant functions pertaining o spin system simulation---in the Python package MagPy. The packages main function, \texttt{lvnsolve}, can evolve a density matrix given an initial condition and a Hamiltonian over a specified list of points in time. 

\subsection{Further Work}

A next step in the expansion of the work of this report would be to adapt the function \texttt{lvnsolve} to handle any matrix-valued function as a Hamiltonian. This would enable more systems than those relevant to NMR or MRI to take advantage of the increased accuracy and efficiency of MagPy. In addition to this, we would like to test scaling and splitting methods in conjunction with the Magnus expansion to further increase the effectiveness of our implementation.

It is also worth investigating further the convergence of different integral approximations. Specifically that more accurate techniques for integration do not necessarily result in faster convergence to our reference solution. We do not expect the midpoint approximation to perform better than the SciPy package's more sophisticated approximation techniques.

Testing of systems with more spins, both interacting and not, would aid in further verifying the effectiveness of our implementation, and potentially find more systems to which our method can be adapted. Alongside this, testing of the speed and number of computations required by \texttt{lvnsolve} would determine the complexity of the function.

\pagebreak

\addcontentsline{toc}{section}{Appendix}
\section*{Appendix}

\addcontentsline{toc}{subsection}{The Kronecker Product}
\subsection*{The Kronecker Product}

    The Kronecker product of $A\in\mathbb{C}^{m\times n}$ and $B\in\mathbb{C}^{p\times q}$ is
    $$A \otimes B = \begin{bmatrix}a_{11}B & \cdots & a_{1n}B \\ \vdots & \ddots & \vdots \\ a_{m1}B & \cdots & a_{mn}B\end{bmatrix}\in \mathbb{C}^{pm \times qn}.$$
\noindent
The mixed-product property of the Kronecker product is that for matrices $A,B,C,D$ such that the products $AC$ and $BD$ can be formed, we have 
$$\left(A\otimes B\right)\left(C\otimes D\right) = AC\otimes BD.$$

\addcontentsline{toc}{subsection}{An Error Bound of Scaling and Squaring}
\subsection*{An Error Bound of Scaling and Squaring}

If $||A|| \leq 2^{j-1},$ then
$$ [R_{qq}(A)]^{2j} = e^{A+E},$$
where $E$ is a matrix such that
$$ \frac{\norm{E}}{\norm{A}} \leq 8\left[\frac{\norm{A}}{2^j}\right]^{2q}\left(\frac{(q!)^2}{(2q)!(2q+1)!}\right).$$
And so, we may choose a $(q,j)$ pair such that
$$\frac{\norm{E}}{\norm{A}} \leq \epsilon,$$
for a chosen error tolerance $\epsilon.$

\bgroup
\def\arraystretch{1.7}
\begin{table}[H]
\centering
\begin{tabular}{|c|c|c|c|c|c|}
    \hline
      \diagbox[innerwidth=1.3cm]{\norm{A}}{$\epsilon$}
                   & $10^{-3}$ & $10^{-6}$ & $10^{-9}$ & $10^{-12}$ & $10^{-15}$ \\
      \hline
      $10^{-2}$ & $(1,0)$ & $(1,0)$ & $(2,0)$ & $(3,0)$ & $(3,0)$\\
      \hline
      $10^{-1}$ & $(1,0)$ & $(2,0)$ & $(3,0)$ & $(4,0)$ & $(4,0)$\\
      \hline
      $10^{0}$ & $(2,1)$ & $(3,1)$ & $(4,1)$ & $(5,1)$ & $(6,1)$\\
      \hline
      $10^{1}$ & $(2,5)$ & $(3,5)$ & $(4,5)$ & $(5,5)$ & $(6,5)$\\
      \hline
      $10^{2}$ & $(2,8)$ & $(3,8)$ & $(4,8)$ & $(5,8)$ & $(6,8)$\\
      \hline
      $10^{3}$ & $(2,11)$ & $(3,11)$ & $(4,11)$ & $(5,11)$ & $(6,11)$\\
      \hline
\end{tabular}
\caption{\label{tab:pade}Optimal values for $(q,j)$ with diagonal Padé.}
\end{table}
\egroup
\noindent
Since $[R_{qq}(A/{2^j})]^{2^j}$ requires about $(q+j+\frac{1}{3})n^3$ flops to evaluate, it makes sense to choose a $(q,j)$ pair such that $q+j$ is minimum \cite{Nineteen}. Table (\ref{tab:pade}) shows the optimal scaling and squaring parameters for a fixed \norm{A} and error tolerance $\epsilon$.

\addcontentsline{toc}{subsection}{The Lanczos Algorithm}
\subsection*{The Lanczos Algorithm}
Given a Hermitian matrix $A$ of size $n \times n$, a column vector $b$ of size $n$, and a number of iterations $m$,
\begin{enumerate}
    \item $v_1 = b / \norm{b}_2$.
    \item $w_1 = Av_1 - \langle Av_1,v_1 \rangle v_1$.
    \item For $j = 2, \ldots, m$, do
    \begin{enumerate}[a)]
        \item $\beta_j = \norm{w_{j-1}}$.
        \item $v_j = w_{j-1} / \beta_j$.
        \item $\alpha_j = \langle Av_j,v_j \rangle$.
        \item $w_j = Av_j - \alpha_j v_j - \beta_j v_{j-1}$.
    \end{enumerate}
    \item V is the matrix with columns $v_1, \dots, v_m$, and 
    $$T = \begin{pmatrix} \alpha_1 &  \beta_2 & & & & 0 \\ \beta_2 & \alpha_2 & \beta_3 & & & \\ & \beta_3 & \alpha_3 & \ddots & & \\ & & \ddots & \ddots & \beta_{m-1} & \\ & & & \beta_{m-1} & \alpha_{m-1} & \beta_m \\ 0 & & & & \beta_m & \alpha_m \end{pmatrix}.$$
\end{enumerate}

\begin{listing}
\begin{minted}{python}
import numpy as np
import scipy as sp

def lanczos(A, b, m=None):
    n = A.shape[0]
    if m is None:
        m = n
    V = np.zeros((n, m), dtype='complex')
    W = np.zeros((n, m), dtype='complex')
    alpha = np.zeros((m, 1), dtype='complex')
    beta = np.zeros((m, 1), dtype='complex')
    
    V[:, 0] = b / np.linalg.norm(b)
    W[:, 0] = A @ V[:, 0]
    alpha[0] = np.vdot(W[:, 0], V[:, 0])
    W[:, 0] = W[:, 0] - alpha[0]*V[:, 0]
    
    for j in range(1, m):
        beta[j] = np.linalg.norm(W[:, j-1])
        V[:, j] = W[:, j-1] / beta[j]
        W[:, j] = A @ V[:, j]
        alpha[j] = np.vdot(W[:, j], V[:, j])
        W[:, j] = W[:, j] - alpha[j]*V[:, j] - beta[j]*V[:, j-1]

    return V, np.diagflat(alpha) + np.diagflat(beta[1:], 1) + np.diagflat(beta[1:], -1)

def krylov_expm(A, b, m=None):
    V, T = lanczos(A, b, m)
    return np.linalg.norm(b) * V @ sp.linalg.expm(T) @ np.eye(1, T.shape[0])[0]
\end{minted}
\caption{Lanczos algorithm and Krylov subspace approximation.}
\end{listing}

\pagebreak

\addcontentsline{toc}{section}{Bibliography}
\renewcommand\refname{Bibliography}
\bibliographystyle{plain}
\bibliography{refs}

\begin{thebibliography}{10}

\bibitem{Higham}
A.~H. Al-Mohy and N.~J. Higham.
\newblock {\em A New Scaling and Squaring Algorithm for the Matrix Exponential}.
\newblock The University of Manchester, 2009.

\bibitem{BlanesMagnus}
S.~Blanes, F.~Casas, J.~A. Oteo, and J.~Ros.
\newblock {\em The Magnus expansion and some of its applications}.
\newblock Physics Reports, 2008.

\bibitem{BroxsonKron}
B.~J. Broxson.
\newblock {\em The Kronecker Product}.
\newblock University of North Florida, 2006.

\bibitem{EverlingExpm}
W.~Everling and M.~L. Liou.
\newblock {\em On the evaluation of $e^{AT}$ by power series}.
\newblock IEEE, 1967.

\bibitem{spins}
M.~Foroozandeh.
\newblock {\em Spin dynamics during chirped pulses: applications to homonuclear decoupling and broadband excitation}.
\newblock Journal of Magnetic Resonance, 2020.

\bibitem{SaadPadeKrylov}
E.~Gallopoulous and Y.~Saad.
\newblock {\em On the paralell solution of parabolic equations}.
\newblock RIACS, 1989.

\bibitem{HochbruckKrylov}
M.~Hochbruck and C.~Lubich.
\newblock {\em On Krylov subspace approximations to the matrix exponential operator}.
\newblock SIAM, 1997.

\bibitem{JacobiGLQ}
C.~G.~J. Jacobi.
\newblock {\em Ueber Gauß neue Methode, die Werthe der Integrale näherungsweise zu finden}.
\newblock Journal für Reine und Angewandte Mathematik, 1826.

\bibitem{QuTiP}
J.~R. Johansson, P.~D. Nation, and F.~Nori.
\newblock {\em QuTiP 2: A Python framework for the dynamics of open quantum systems}.
\newblock Computer Physics Communications, 2013.

\bibitem{Lanczos}
C.~Lanczos.
\newblock {\em An Iteration Method for the Solution of the Eigenvalue Problem of Linear Differential and Integral Operators}.
\newblock Journal of Research of the National Bureau of Standards, 1950.

\bibitem{LiouExpm}
M.~L. Liou.
\newblock {\em A Novel Method of Evaluating Transient Response}.
\newblock IEEE, 1966.

\bibitem{MacedoVec}
H.~D. Macedo and J.~N. Oliveira.
\newblock {\em Typing Linear Algebra: A Biproduct-oriented Approach}.
\newblock Science of Computer Programming, 2013.

\bibitem{Magnus}
W.~Magnus.
\newblock {\em On the Exponential Solution of Differential Equations for a Linear Operator}.
\newblock Communications on Pure and Applied Mathematics, 1954.

\bibitem{Mazzi}
G.~Mazzi.
\newblock {\em Numerical Treatment of the Liouville-von Neumann Equation for Quantum Spin Dynamics}.
\newblock University of Edinburgh, 2010.

\bibitem{su}
W.~Pfeifer.
\newblock {\em The Lie Algebra su(N)}.
\newblock Birkh\"auser, Basel, 2003.

\bibitem{spins2}
J.~Randall, A.~M. Lawrence, S.~C. Webster, S.~Weidt, N.~V. Vitanov, and W.~K. Hensinger.
\newblock {\em Generation of high-fidelity quantum control methods for multilevel systems}.
\newblock Physical Review A, 2018.

\bibitem{Suli}
E.~S\"uli and D.~F. Mayers.
\newblock {\em An Introduction to Numerical Analysis}.
\newblock Cambridge University Press, 2003.

\bibitem{Assche}
W.~Van~Assche.
\newblock {\em Padé and Hermite-Padé approximation and orthogonality}.
\newblock Surveys in Approximation Theory, 2006.

\bibitem{Nineteen}
C.~Van~Loan and C.~Moler.
\newblock {\em Nineteen Dubious Ways to Compute the Exponential of a Matrix, Twenty-Five Years Later}.
\newblock SIAM Review, 2003.

\bibitem{VargaPade}
R.~S. Varga.
\newblock {\em On higher order stable implicit methods for solving parabolic partial differential equations}.
\newblock Journal of Mathematical Physics, 1961.

\bibitem{SciPy}
P.~Virtanen et~al.
\newblock {\em {{SciPy} 1.0: Fundamental Algorithms for Scientific Computing in Python}}.
\newblock Nature Methods, 2020.

\end{thebibliography}

\end{document}